\documentclass[12pt]{article}

\usepackage{amsbsy, amstext, amssymb, amsthm, amsmath,bm,wasysym}
\usepackage{graphicx,booktabs,natbib,hyperref,rotating}
\usepackage[margin=1in]{geometry}
\usepackage[lined,boxed,linesnumbered]{algorithm2e}
\DontPrintSemicolon
\usepackage[T1]{fontenc}

\usepackage{xr}
\newtheorem{proposition}{{\bf Proposition}}
\newtheorem{lemma}{{\bf Lemma}}

\newtheorem{assumption}{{\bf Assumption}}

\newtheorem{theorem}{{\bf Theorem}}

\newcommand{\E}{\mathrm{E}}

\newcommand{\Cov}{\mathrm{Cov}}

\newcommand{\vect}{\mathrm{vec}}

\newcommand{\tr}{\mathrm{tr}}

\newcommand{\rank}{\mathrm{rank}}

\newcommand{\diag}{\mathrm{diag}}

\newcommand{\Var}{\mathrm{Var}}
\def\amp{\mathop{\;\:}\nolimits}

\newcommand{\Abf}{{\bm A}}
\newcommand{\Bbf}{{\bm B}}
\newcommand{\Cbf}{{\bm C}}
\newcommand{\Dbf}{{\bm D}}
\newcommand{\Ebf}{{\bm E}}

\newcommand{\Gbf}{{\bm G}}
\newcommand{\Hbf}{{\bm H}}
\newcommand{\Ibf}{{\bm I}}

\newcommand{\Lbf}{{\bm L}}

\newcommand{\Mbf}{{\bm M}}
\newcommand{\Nbf}{{\bm N}}

\newcommand{\Pbf}{{\bm P}}

\newcommand{\Rbf}{{\bm R}}
\newcommand{\Sbf}{{\bm S}}

\newcommand{\Ubf}{{\bm U}}
\newcommand{\Vbf}{{\bm V}}

\newcommand{\Xbf}{{\bm X}}
\newcommand{\Ybf}{{\bm Y}}
\newcommand{\Zbf}{{\bm Z}}

\newcommand{\dbf}{{\bm d}}

\newcommand{\rbf}{{\bm r}}

\newcommand{\ubf}{{\bm u}}
\newcommand{\vbf}{{\bm v}}

\newcommand{\xbf}{{\bm x}}
\newcommand{\ybf}{{\bm y}}

\newcommand{\zerobf}{{\mathbf 0}}
\newcommand{\onebf}{{\mathbf 1}}

\newcommand{\greekbold}[1]{\mbox{\boldmath $#1$}}
\newcommand{\alphabf}{\greekbold{\alpha}}
\newcommand{\betabf}{\greekbold{\beta}}

\newcommand{\epsilonbf}{\greekbold{\epsilon}}

\newcommand{\gammabf}{\greekbold{\gamma}}
\newcommand{\mubf}{\greekbold{\mu}}

\newcommand{\sigmabf}{\greekbold{\sigma}}

\newcommand{\thetabf}{\greekbold{\theta}}
\newcommand{\lambdabf}{\greekbold{\lambda}}
\newcommand{\Deltabf}{\greekbold{\Delta}}
\newcommand{\Gammabf}{\greekbold{\Gamma}}
\newcommand{\Phibf}{\greekbold{\Phi}}
\newcommand{\Psibf}{\greekbold{\Psi}}

\newcommand{\Lambdabf}{\greekbold{\Lambda}}
\newcommand{\Omegabf}{\greekbold{\Omega}}

\title{MM Algorithms for Variance Components Models}
\author{Hua Zhou\\
Department of Biostatistics\\
University of California, Los Angeles\\
huazhou@ucla.edu\\
\and
Liuyi Hu\\
Department of Statistics\\
North Carolina State University\\
lhu@ncsu.edu\\
\and
Jin Zhou\\
Division of Epidemiology \\and Biostatistics\\
University of Arizona\\
jzhou@email.arizona.edu\\
\and
Kenneth Lange\\
Departments of Biomathematics, \\Human Genetics, and Statistics\\
University of California, Los Angeles\\
klange@ucla.edu
}
\date{}

\begin{document}
\maketitle

\begin{abstract}
Variance components estimation and mixed model analysis are central themes in statistics with applications in numerous scientific disciplines. Despite the best efforts of generations of statisticians and numerical analysts, maximum likelihood estimation and restricted maximum likelihood estimation of variance component models remain numerically challenging.  Building on the minorization-maximi\-za\-tion (MM) principle, this paper presents a novel iterative algorithm for variance components estimation. MM algorithm is trivial to implement and competitive on large data problems. The algorithm readily extends to more complicated problems such as linear mixed models, multivariate response models possibly with missing data, maximum a posteriori estimation, penalized estimation, and generalized estimating equations (GEE).  We establish the global convergence of the MM algorithm to a KKT point and demonstrate, both numerically and theoretically, that it converges faster than the classical EM algorithm when the number of variance components is greater than two and all covariance matrices are positive definite.

\end{abstract}

\section{Introduction}

Variance components and linear mixed models are among the most potent tools in a statistician's toolbox. They are essential topics in graduate-level linear model courses and the subject of many
current papers and research monographs \citep{RaoLKleffe88VarCompBook,SearleCasellaMcCulloch92VarCompBook,Rao97VarCompBook,KhuriMathewSinha98LMMBook,Demidenko13MixedModelsBook}. Their applications in agriculture, biology, economics, genetics, epidemiology, and medicine are too numerous to cover here in detail. The recommended books  \citep{VerbekeMolenberghs00LMMBook, Weiss05Book,FitzmauriceLairdWare11Book}
stress longitudinal data analysis.

Given an observed $n \times 1$ response vector $\ybf$ and $n \times p$ predictor matrix $\Xbf$, the simplest variance components model postulates that $\Ybf \sim N(\Xbf \betabf, \Omegabf)$, where
\begin{eqnarray*}
	\Omegabf & = & \sum_{i=1}^m \sigma_i^2 \Vbf_i ,
\end{eqnarray*}
and the $\Vbf_1,\ldots,\Vbf_m$ are $m$ fixed positive semidefinite matrices. The parameters of the 
model can be divided into mean effects $(\beta_1,\ldots,\beta_p)$  and variance components 
$(\sigma_1^2,\ldots,\sigma_m^2)$, summarized by vectors $\betabf$ and $\sigmabf^2$. Throughout 
we assume $\Omegabf$ is positive definite. The extension to singular $\Omegabf$
will not be pursued here. Estimation revolves around the log-likelihood function
\begin{eqnarray}
	L(\betabf, \sigmabf^2) 
	&=&  - \frac 12 \ln \det \Omegabf - \frac{1}{2} (\ybf - \Xbf \betabf)^T \Omegabf^{-1} (\ybf - \Xbf \betabf). \label{eqn:vc-loglike}
\end{eqnarray}
Among the commonly used methods for estimating variance components, maximum likelihood estimation (MLE) \citep{HartleyRao67MLE} and restricted (or residual) MLE (REML) \citep{Harville77REML} are the most popular. REML first projects $\ybf$ to the null space of $\Xbf$ and then estimates variance components based on the projected responses. If the columns of the matrix $\Bbf$ span the null space of
$\Xbf^T$, then REML estimates the $\sigma_i^2$ by maximizing the log-likelihood of the redefined response
vector $\Bbf^T \Ybf$, which is normally distributed with mean $\zerobf$ and covariance $\Bbf^T \Omegabf \Bbf = \sum_{i=1}^m \sigma_i^2 \Bbf^T \Vbf_i \Bbf$. 

There exists a large literature on iterative algorithms for finding MLE and REML  \citep{LairdWare82LMM,LindstromBates88NewtonEMLMM,LindstromBates90NonlinearLMM,HarvilleCallanan90REMLAlgo,CallananHarville91REMLAlgo,BatesPinheiro98Multilevel,SchaferYucel02MLMM}. Fitting variance component models remains a challenge in models with a large sample size $n$ or a large number of variance components $m$. Newton's method \citep{LindstromBates88NewtonEMLMM} converges quickly but is numerically unstable owing to the non-concavity of the log-likelihood. Fisher's scoring algorithm replaces the observed information matrix in Newton's method by the
expected information matrix and yields an ascent algorithm when safeguarded by step halving. However the calculation and inversion of expected information matrices cost $O(mn^3) + O(m^3)$ flops for unstructured $\Vbf_i$ and quickly become impractical when either $n$ or $m$ is large. The expectation-maximization (EM) algorithm initiated by Dempster et al.\ is a third alternative 
\citep{Dempster1977,LairdWare82LMM,LairdLangeStram87LMMEM,LindstromBates88NewtonEMLMM,BatesPinheiro98Multilevel}.  Compared to Newton's method, the EM algorithm is easy to implement and numerically stable, but painfully slow to converge. In practice, a strategy of priming Newton's method by a few EM steps leverages the stability of EM and the faster convergence of second-order methods. Quasi-Newton methods dispense with explicit calculation of the observed information while achieving a superlinear rate of convergence.

In this paper we derive a minorization-maximization (MM) algorithm for finding the MLE and REML estimates of variance components. We prove global convergence of the MM algorithm to a Karush-Kuhn-Tucker (KKT) point and explain why MM generally converges faster than EM for models with more than two variance components. We also sketch extensions of the MM algorithm to the multivariate response model with possibly missing responses, the linear mixed model (LMM), maximum a posteriori (MAP) estimation, penalized estimation, and generalized estimating equations (GEE). The numerical efficiency of the MM algorithm is illustrated through simulated data sets and a genomic example with more than 200 variance components.

\section{Preliminaries}

\subsection*{Background on MM algorithms}

Throughout we reserve Greek letters for parameters and indicate the current iteration number by a superscript $t$. The MM principle for maximizing an objective function $f(\thetabf)$ involves minorizing the objective function $f(\thetabf)$ by a surrogate function $g(\thetabf \mid \thetabf^{(t)})$ around the current iterate $\thetabf^{(t)}$ of a search \citep{Lange00OptTrans}.  Minorization is defined by the
two conditions
\begin{eqnarray}
f(\thetabf^{(t)}) & = & g(\thetabf^{(t)} \mid \thetabf^{(t)})  \label{majorization_definition} \\
f(\thetabf) & \ge & g(\thetabf \mid \thetabf^{(t)})\: , \quad \quad \thetabf \ne \thetabf^{(t)} . \nonumber
\end{eqnarray}
In other words, the surface $\thetabf \mapsto g(\thetabf \mid \thetabf^{(t)})$ lies below the
surface $\thetabf \mapsto f(\thetabf)$ and is tangent to it at the point $\thetabf=\thetabf^{(t)}$.  Construction of the minorizing function $g(\thetabf \mid \thetabf^{(t)})$ constitutes the first
M of the MM algorithm. The second M of the algorithm maximizes the surrogate $g(\thetabf \mid \thetabf^{(t)})$
rather than $f(\thetabf)$.  The point $\thetabf^{(t+1)}$ maximizing 
$g(\thetabf \mid \thetabf^{(t)})$ satisfies the ascent property 
$f(\thetabf^{(t+1)}) \ge f(\thetabf^{(t)})$.  This fact follows from the inequalities
\begin{eqnarray}
f(\thetabf^{(t+1)}) & \ge &  g(\thetabf^{(t+1)} \mid \thetabf^{(t)}) \amp \ge \amp g(\thetabf^{(t)} \mid \thetabf^{(t)}) \amp = \amp f(\thetabf^{(t)}), \label{eqn:MM-monotonicity}
\end{eqnarray}
reflecting the definition of $\thetabf^{(t+1)}$ and the tangency and domination conditions (\ref{majorization_definition}). The ascent property makes the MM algorithm remarkably stable.  The validity of the descent property depends only on increasing $g(\thetabf \mid \thetabf^{(t)})$, not on maximizing $g(\thetabf \mid \thetabf^{(t)})$. With obvious changes, the MM algorithm also applies to minimization rather than to maximization. To minimize a function $f(\thetabf)$, we majorize it by a surrogate function $g(\thetabf \mid \thetabf^{(t)})$ and minimize $g(\thetabf \mid \thetabf^{(t)})$ to produce the next iterate $\thetabf^{(t+1)}$. The acronym should not be confused with the maximization-maximization algorithm in the variational Bayes context \citep{Jeon12Thesis}.

\begin{figure}
\begin{center}
\includegraphics[width=4.5in]{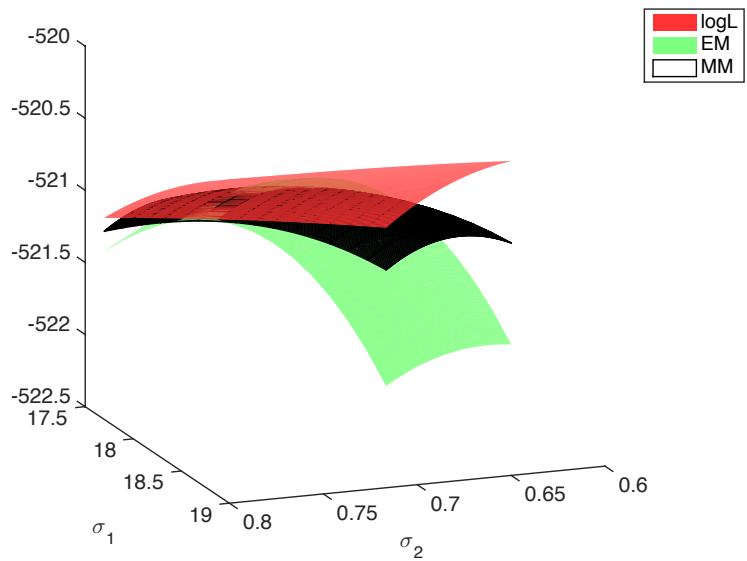}
\end{center}
\caption{Log-likelihood surface of a 2-variance component model and the surrogate functions of EM and MM minorizing the objective function at point ($\sigma_1^{2(t)}, \sigma_2^{2(t)})=(18.5, 0.7)$.}
\label{fig:MM-vs-EM}
\end{figure}

The MM principle \citep{DeLeeuw994block,Heiser1995convergent,Kiers2002setting,Lange00OptTrans,HunterLange04MMTutorial,WuLange10EMMM} finds applications in multidimensional scaling \citep{BorgGroenen2005modern},
ranking of sports teams \citep{Hunter04BreadleyTerry}, variable selection \citep{HunterLi05MM}, optimal experiment design \citep{Yu10OptDesign}, multivariate statistics \citep{ZhouLange10DMMLE}, geometric programming \citep{LangeZhou14GPMM}, and many other areas \citep[Chapter 12]{Lange10NumAnalBook}. The celebrated EM principle \citep{Dempster1977} is a special case of the MM principle. The $Q$ function produced in the E step of an EM algorithm minorizes the log-likelihood up to an irrelevant constant. Thus, both EM and MM share the same advantages: simplicity, stability, graceful adaptation to constraints, and the tendency to avoid large matrix inversion. The more general MM perspective frees algorithm derivation from the missing data straitjacket and invites wider applications \citep{WuLange10EMMM}. Figure \ref{fig:MM-vs-EM} shows the minorization functions of EM and MM for a variance components model with $m=2$ variance components.

EM and MM algorithms often exhibit slow convergence. Fortunately, this defect can be remedied by off-the-shelf acceleration techniques for fixed point iterations. The recently developed squared iterative method (SQUAREM) \citep{Varadhan08SQUAREM} and the quasi-Newton acceleration method \citep{ZhouAlexanderLange09QN} are particularly attractive, given their simplicity and minimal memory and computational costs. Our numerical experiments feature the unadorned MM algorithm and the quasi-Newton accelerated MM (aMM) algorithm based on one secant pair. Using more secant pairs is likely to further improve performance.

\subsection*{Convex matrix functions}

For symmetric matrices we write $\Abf \preceq \Bbf$ when $\Bbf - \Abf$ is positive semidefinite and $\Abf \prec \Bbf$ if $\Bbf - \Abf$ is positive definite. A matrix-valued function $f$ is said to be (matrix) convex if 
\begin{eqnarray*}
	f[\lambda \Abf + (1-\lambda) \Bbf] & \preceq & \lambda f(\Abf) + (1-\lambda) f(\Bbf)
\end{eqnarray*}
for all $\Abf$, $\Bbf$, and $\lambda \in [0,1]$. Our derivation of the MM variance components algorithm hinges on the convexity of the two functions mentioned in the next lemma.
\begin{lemma}
\label{convexity_lemma}
(a) The matrix fractional function $f(\Abf, \Bbf) = \Abf^T \Bbf^{-1} \Abf$ is jointly convex in the
$m \times n$ matrix $\Abf$ and the $m \times m$ positive definite matrix $\Bbf$.
(b) The log determinant function $f(\Bbf) = \ln \det \Bbf$ is concave on the set of positive definite matrices.
\end{lemma}
\begin{proof}
The matrix fractional function is matrix convex because its epigraph
\begin{eqnarray*}
	\{(\Abf, \Bbf, \Cbf): \Bbf \succ \zerobf, \Abf^T \Bbf^{-1} \Abf \preceq \Cbf\}
	& = & \left\{ (\Abf, \Bbf, \Cbf): \Bbf \succ \zerobf, \begin{pmatrix} \Bbf & \Abf \\ \Abf^T & \Cbf \end{pmatrix} \succeq \zerobf \right\}
\end{eqnarray*}
is a convex set. Here $\Cbf$ varies over the set of $n \times n$ positive semidefinite matrices.
The equivalence of these two epigraph representations is proved in \citet[A.5.5]
{BoydVandenberghe04Book}. For the concavity of the log determinant, see \citet[p74]{BoydVandenberghe04Book}.
\end{proof}

\section{Univariate response model}
\label{sec:univ-y}

Our strategy for maximizing the log-likelihood \eqref{eqn:vc-loglike} is to alternate updating the mean parameters $\betabf$ and the variance components $\sigmabf^2$. Updating $\betabf$ given $\sigmabf^2$ is a standard general least squares problem with solution
\begin{eqnarray}
\betabf^{(t+1)} & = & (\Xbf^T \Omegabf^{-(t)}\Xbf)^{-1}\Xbf^T \Omegabf^{-(t)}\ybf.
\label{generalized_least_squares}
\end{eqnarray}
Updating $\sigmabf^2$ given $\betabf^{(t)}$ depends on two minorizations. If we assume that all of the $\Vbf_i$ are positive definite,
then the joint convexity of the map $(\Xbf, \Ybf) \mapsto \Xbf^T \Ybf^{-1} \Xbf$ 
for positive definite $\Ybf$ implies that
\begin{eqnarray*}
 \Omegabf^{(t)} \Omegabf^{-1} \Omegabf^{(t)} 
	&=& \left( \sum_{i=1}^m \sigma_i^{2(t)} \Vbf_i \right)  \left( \sum_{i=1}^m \sigma_i^2 \Vbf_i \right)^{-1} \left( \sum_{i=1}^m \sigma_i^{2(t)} \Vbf_i \right) \nonumber \\
	&=&  \left( \sum_{i=1}^m \frac{\sigma_i^{2(t)}}{\sum_{j} \sigma_{j}^{2(t)}} \frac{\sum_{j} \sigma_{j}^{2(t)}}{\sigma_i^{2(t)}}  \sigma_i^{2(t)} \Vbf_i \right) \cdot \left( \sum_{i=1}^m \frac{\sigma_i^{2(t)}}{\sum_{j} \sigma_{j}^{2(t)}} \frac{\sum_{j} \sigma_{j}^{2(t)}}{\sigma_i^{2(t)}} \sigma_i^2 \Vbf_i \right)^{-1} \nonumber \\
	& & \cdot \left( \sum_{i=1}^m \frac{\sigma_i^{2(t)}}{\sum_{j} \sigma_{j}^{2(t)}} \frac{\sum_{j} \sigma_{j}^{2(t)}}{\sigma_i^{2(t)}}  \sigma_i^{2(t)} \Vbf_i \right) \nonumber \\
	&\preceq& \sum_{i=1}^m \frac{\sigma_i^{2(t)}}{\sum_{j} \sigma_{j}^{2(t)}} \left( \frac{\sum_{j} \sigma_{j}^{2(t)}}{\sigma_i^{2(t)}} \sigma_i^{2(t)} \Vbf_i \right) \left( \frac{\sum_{j} \sigma_{j}^{2(t)}}{\sigma_i^{2(t)}} \sigma_i^{2} \Vbf_i \right)^{-1} \left( \frac{\sum_{j} \sigma_{j}^{2(t)}}{\sigma_i^{2(t)}} \sigma_i^{2(t)} \Vbf_i \right)  \\
	&=& \sum_{i=1}^m \frac{\sigma_i^{4(t)}}{\sigma_i^2} \Vbf_i \Vbf_i^{-1} \Vbf_i \nonumber \\
	&=& \sum_{i=1}^m \frac{\sigma_i^{4(t)}}{\sigma_i^2} \Vbf_i. \nonumber
\end{eqnarray*}
When one or more of the $\Vbf_i$ are rank deficient, we replace each $\Vbf_i$ by 
$\Vbf_{i,\epsilon} = \Vbf_i + \epsilon \Ibf$ for $\epsilon>0$ small. 
Sending $\epsilon$ to 0 in the just proved majorization 
\begin{eqnarray*}
\Omegabf_\epsilon^{(t)} \Omegabf_\epsilon^{-1} \Omegabf_\epsilon^{(t)} & \preceq &
	\sum_{i=1}^m \frac{\sigma_i^{4(t)}}{\sigma_i^2} \Vbf_{i,\epsilon}
\end{eqnarray*}
gives the desired majorization
\begin{eqnarray*}
\Omegabf^{(t)} \Omegabf^{-1} \Omegabf^{(t)} & \preceq & \sum_{i=1}^m \frac{\sigma_i^{4(t)}}{\sigma_i^2} \Vbf_{i} 
\end{eqnarray*}
in the general case. Negating both sides leads to the minorization 
\begin{eqnarray}
	- (\ybf - \Xbf \betabf)^T \Omegabf^{-1} (\ybf - \Xbf \betabf) & \succeq & - (\ybf - \Xbf \betabf)^T \Omegabf^{-(t)} \left( \sum_{i=1}^m \frac{\sigma_i^{4(t)}}{\sigma_i^2} \Vbf_{i} \right) \Omegabf^{- (t)} (\ybf - \Xbf \betabf) \label{eqn:matrix-ineq}
\end{eqnarray}
that effectively separates the variance components $\sigma_1^2, \ldots, \sigma_m^2$ in the quadratic term of the log-likelihood \eqref{eqn:vc-loglike}. 

The convexity of the function $\Abf \mapsto - \log \det \Abf$ is equivalent
to the supporting hyperplane minorization 
\begin{eqnarray}
-\ln \det \Omegabf & \ge & -\ln \det \Omegabf^{(t)} - \tr [\Omegabf^{-(t)} (\Omegabf - \Omegabf^{(t)})] \label{eqn:logdet-majorization}
\end{eqnarray}
that separates $\sigma_1^2, \ldots, \sigma_m^2$ in the log determinant term of the log-likelihood \eqref{eqn:vc-loglike}. Combination of the minorizations \eqref{eqn:matrix-ineq} and \eqref{eqn:logdet-majorization} gives the overall minorization
\begin{eqnarray}
	& & g(\sigmabf^2 \mid \sigmabf^{2(t)}) \nonumber \\
	&=& - \frac 12 \tr (\Omegabf^{-(t)} \Omegabf) - \frac 12 (\ybf - \Xbf \betabf^{(t)})^T \Omegabf^{-(t)} \left( \sum_{i=1}^m \frac{\sigma_i^{4(t)}}{\sigma_i^2} \Vbf_i \right) \Omegabf^{-(t)} (\ybf - \Xbf \betabf^{(t)}) + c^{(t)} \label{g_def} \\
	&=& \sum_{i=1}^m \left[- \frac{\sigma_i^2}{2} \tr (\Omegabf^{-(t)}  \Vbf_i) - \frac 12 \frac{\sigma_i^{4(t)}}{\sigma_i^2} (\ybf - \Xbf \betabf^{(t)})^T \Omegabf^{-(t)} \Vbf_i \Omegabf^{-(t)} (\ybf - \Xbf \betabf^{(t)}) \right] + c^{(t)}, \nonumber
\end{eqnarray}
where $c^{(t)}$ is an irrelevant constant.
Maximization of $g(\sigmabf^2 \mid \sigmabf^{2(t)})$ with respect to $\sigma_i^2$ yields the lovely multiplicative update
\begin{eqnarray}
	\sigma_i^{2(t+1)} & = & \sigma_i^{2(t)} \sqrt{\frac{(\ybf - \Xbf \betabf^{(t)})^T \Omegabf^{-(t)} \Vbf_i \Omegabf^{-(t)} (\ybf - \Xbf \betabf^{(t)})}{\tr (\Omegabf^{-(t)}  \Vbf_i)}}, \quad i = 1,\ldots, m. \label{eqn:MM-update-univariate}
\end{eqnarray}
To preserve the uniqueness and continuity of the algorithm map, 
we must take $\sigma_i^{2(t+1)} = 0$ whenever $\sigma_i^{2(t)} = 0$. 
As a sanity check on our derivation, consider the partial derivative
\begin{eqnarray}
\frac{\partial}{\partial \sigma_i^2} L(\betabf, \sigmabf^2) & = &
-\frac{1}{2}\tr (\Omegabf^{-1}  \Vbf_i)
+\frac{1}{2}(\ybf - \Xbf \betabf)^T \Omegabf^{-1} \Vbf_i \Omegabf^{-1} (\ybf - \Xbf \betabf). \label{eqn:L-gradient}
\end{eqnarray}
Given $\sigma_i^{2(t)}>0$, it is clear from the update formula \eqref{eqn:MM-update-univariate} that $\sigma_i^{2(t+1)} < \sigma_i^{2(t)}$ when $\frac{\partial}{\partial \sigma_i^2} L<0$.
Conversely $\sigma_i^{2(t+1)} > \sigma_i^{2(t)}$ when $\frac{\partial}{\partial \sigma_i^2} L>0$. Algorithm \ref{algo:varcomp-univariate-y} summarizes the MM algorithm for MLE of the univariate response model \eqref{eqn:vc-loglike}.

\begin{algorithm}[ht!]
\SetKwInOut{Input}{Input}\SetKwInOut{Output}{Output}
\Input{$\ybf$, $\Xbf$, $\Vbf_1, \ldots, \Vbf_m$}
\Output{MLE $\hat \betabf$, $\hat \sigma_1^2, \ldots, \hat \sigma_m^2$}
Initialize $\sigma_i^{(0)}>0$, $i=1,\ldots,m$ \;
\Repeat{objective value converges}{
$\Omegabf^{(t)} \gets \sum_{i=1}^m \sigma_i^{2(t)} \Vbf_i$ \;
$\betabf^{(t)} \gets \arg \min_{\betabf} \, (\ybf - \Xbf \betabf)^T \Omegabf^{-(t)} (\ybf - \Xbf \betabf)$ 
$\sigma_i^{2(t+1)} \gets \sigma_i^{2(t)} \sqrt{\frac{(\ybf - \Xbf \betabf^{(t)})^T \Omegabf^{-(t)} \Vbf_i \Omegabf^{-(t)} (\ybf - \Xbf \betabf^{(t)})}{\tr (\Omegabf^{-(t)}  \Vbf_i)}}, \quad i = 1,\ldots, m$ \;
}
\caption{MM algorithm for MLE of the variance components of model \eqref{eqn:vc-loglike}.}
\label{algo:varcomp-univariate-y}
\end{algorithm}

The update formula (\ref{eqn:MM-update-univariate}) assumes that the
numerator under the square root sign is nonnegative and the denominator is
positive. The numerator requirement is a consequence of the positive
semidefiniteness of $\Vbf_i$. The denominator requirement can be verified
through the Hadamard (elementwise) product representation 
$\tr(\Omegabf^{-(t)}\Vbf_i) = \onebf^T (\Omegabf^{-(t)} \odot \Vbf_i) \onebf$.
The following lemma of \citet{Schur11Hadamard} is crucial. We give a self-contained probabilistic proof.
\begin{lemma}[Schur]
The Hadamard product of a positive definite matrix with a positive semidefinite matrix with positive diagonal entries is positive definite. 
\end{lemma}
\begin{proof} 
Let $\Xbf = (X_1, \ldots, X_n)^T$ be a random normal vector with mean $\zerobf$ and positive definite covariance matrix $\Abf$. Let $\Ybf = (Y_1, \ldots, Y_n)^T$ be a random normal vector 
independent of $\Xbf$ with mean $\zerobf$ and positive semidefinite covariance matrix $\Bbf$ having positive diagonal entries. Then $\Zbf = \Xbf \odot \Ybf$ has covariances $\E (Z_i Z_j) = \E (X_iY_i X_jY_j) = \E (X_iX_j) \E(Y_i Y_j) = a_{ij} b_{ij}$. It follows that $\Cov(\Zbf) = \Abf \odot \Bbf$. To show $\Abf \odot \Bbf$ is positive definite, suppose on the contrary that 
$\vbf^T (\Abf \odot \Bbf) \vbf = \Var(\vbf^T \Zbf) = 0$ for some $\vbf \ne \zerobf$. Then
\begin{eqnarray*}
	0 &\!\!=\!\!& \Var(\vbf^T \Zbf) = \E \Big(\sum_i v_i X_i Y_i\Big)^2 = \E \Big[ \Big(\sum_i v_i X_i Y_i\Big)^2 \mid \Ybf \Big] = \E [(\vbf \odot \Ybf)^T \Abf (\vbf \odot \Ybf)]
\end{eqnarray*}
implies $\vbf \odot \Ybf = \zerobf$ with probability 1. Since $\vbf \ne \zerobf$, $Y_i = 0$ with probability 1 for some $i$. This contradicts the assumption $b_{ii} = \Var (Y_i) > 0$ for all $i$. 
\end{proof}

We can now obtain the following characterization of the MM iterates.
\begin{proposition}
\label{prop:nonzero-denom}
Assume $\Vbf_i$ has strictly positive diagonal entries. Then $\tr(\Omegabf^{-(t)} \Vbf_i) > 0$ for all $t$. Furthermore if $\sigma_i^{2(0)} > 0$ and $\Omegabf^{-(t)} (\ybf - \Xbf \betabf^{(t)}) \notin \text{null}(\Vbf_i)$ for all $t$, then $\sigma_i^{2(t)} > 0$ for all $t$. When $\Vbf_i$ is
positive definite, $\sigma_i^{2(t)} > 0$ holds if and only if $\ybf \ne \Xbf \betabf^{(t)}$.
\end{proposition}
\begin{proof}
The first claim follows easily from Schur's lemma. The second claim follows by induction.
The third claim follows from the observation that $\text{null}(\Vbf_i) = \{\zerobf\}$.
\end{proof}

In most applications, $\Vbf_m = \Ibf$. Proposition \ref{prop:nonzero-denom} guarantees that if
$\sigma_m^{2(0)}>0$ and the residual vector $\ybf - \Xbf \betabf^{(t)}$ is nonzero, then
$\sigma_m^{2(t)}$ remains positive and thus $\Omegabf^{(t)}$ remains positive definite throughout all iterations. This fact does not prevent any of the sequences $\sigma_i^{2(t)}$ from converging to 0. In this sense, the MM algorithm acts like an interior point method, approaching the optimum from inside the feasible region.

\subsection*{Univariate response: two variance components}

\begin{algorithm}[th!]
\SetKwInOut{Input}{Input}\SetKwInOut{Output}{Output}
\Input{$\ybf$, $\Xbf$, $\Vbf_1, \Vbf_2$}
\Output{MLE $\hat \betabf$, $\hat \sigma_1^2, \hat \sigma_2^2$}
Simultaneous congruence decomposition: $(\Dbf, \Ubf) \gets (\Vbf_1, \Vbf_2)$ \;
Transform data: $\tilde \ybf \gets \Ubf^T \ybf$, $\tilde \Xbf \gets \Ubf^T \Xbf$ \;
Initialize $\sigma_1^{(0)}, \sigma_1^{(0)} > 0$ \;
\Repeat{objective value converges}{
$w_i^{(t)} \gets (\sigma_1^{2(t)} d_i + \sigma_2^{2(t)})^{-1}, \quad i=1,\ldots,n$ \;
$\betabf^{(t)} \gets \arg \min_{\betabf} \, \sum_{i=1}^n w_i^{(t)} (\tilde y_i - \tilde \xbf_i^T \betabf)^2$ 
$\sigma_1^{2(t+1)} \gets \sigma_1^{2(t)} \sqrt{\frac{(\tilde \ybf - \tilde \Xbf \betabf^{(t)})^T (\sigma_1^{2(t)} \Dbf + \sigma_2^{2(t)} \Ibf)^{-1} \Dbf (\sigma_1^{2(t)} \Dbf + \sigma_2^{2(t)} \Ibf)^{-1} (\tilde \ybf - \tilde \Xbf \betabf^{(t)})}{\tr [(\sigma_1^{2(t)} \Dbf + \sigma_2^{2(t)} \Ibf)^{-1} \Dbf]}}$ \;
$\sigma_2^{2(t+1)} \gets \sigma_2^{2(t)} \sqrt{\frac{(\tilde \ybf - \tilde \Xbf \betabf^{(t)})^T (\sigma_1^{2(t)} \Dbf + \sigma_2^{2(t)} \Ibf)^{-2} (\tilde \ybf - \tilde \Xbf \betabf^{(t)})}{\tr [(\sigma_1^{2(t)} \Dbf + \sigma_2^{2(t)} \Ibf)^{-1}]}}$ \;
}
\caption{Simplified MM algorithm for MLE of model \eqref{eqn:vc-loglike} with $m=2$ variance components and $\Omegabf = \sigma_1^2 \Vbf_1 + \sigma_2^2 \Vbf_2$.}
\label{algo:varcomp-univariate-2vc}
\end{algorithm}

The major computational cost of Algorithm 1 is inversion of the covariance matrix $\Omegabf^{(t)}$ at each iteration. Problem specific structures such as block diagonal matrices or a diagonal matrix plus a low-rank matrix are often exploited to speed up matrix inversion. The special case of $m=2$ variance components deserves attention as repeated matrix inversion can be avoided by invoking the 
simultaneous congruence decomposition for two symmetric matrices, one of which is positive definite
\citep{Rao2009linear, HornJohnson85}. This decomposition is also called the generalized eigenvalue decomposition \citep{Golub96Book,BoydVandenberghe04Book}. If one assumes $\Omegabf = \sigma_1^2 \Vbf_1 + \sigma_2^2 \Vbf_2$ and lets $(\Vbf_1, \Vbf_2) \mapsto (\Dbf, \Ubf)$ be the decomposition with $\Ubf$ nonsingular, $\Ubf^T \Vbf_1 \Ubf = \Dbf$ diagonal, and $\Ubf^{T} \Vbf_2 \Ubf =  \Ibf$, then 
\begin{eqnarray}
\Omegabf^{(t)} &=& \Ubf^{-T} (\sigma_1^{2(t)} \Dbf + \sigma_2^{2(t)} \Ibf_n) \Ubf^{-1} \nonumber \\
\Omegabf^{-(t)} &=& \Ubf (\sigma_1^{2(t)} \Dbf + \sigma_2^{2(t)} \Ibf_n)^{-1} \Ubf^T  \nonumber \\
\det(\Omegabf^{(t)}) &=& \det  (\sigma_1^{2(t)} \Dbf + \sigma_2^{2(t)} \Ibf_n) \det (\Ubf^{-T} \Ubf^{-1}) \label{mat_decomposition}\\
&=& \det  (\sigma_1^{2(t)} \Dbf + \sigma_2^{2(t)} \Ibf_n) \det (\Vbf_2). \nonumber
\end{eqnarray}
With the revised responses $\tilde \ybf = \Ubf^T \ybf$ and the revised predictor matrix 
$\tilde \Xbf = \Ubf^T \Xbf$, the update \eqref{eqn:MM-update-univariate} requires 
only vector operations and costs $O(n)$ flops. Updating the fixed effects is a weighted least squares problem with the transformed data $(\tilde \ybf, \tilde \Xbf)$ and observation weights $w_i^{(t)} = (\sigma_1^{2(t)} d_i + \sigma_2^{2(t)})^{-1}$. Algorithm \ref{algo:varcomp-univariate-2vc} summarizes the simplified MM algorithm for two variance components.

\subsection*{Numerical experiments}

This section compares the numerical performance of MM, quasi-Newton accelerated MM, EM, and Fisher scoring on simulated data from a two-way ANOVA random effects model and a genetic model. For ease of comparison, all algorithm runs start from $\sigmabf^{2(0)} = \onebf$ and terminate when the relative change $(L^{(t+1)} - L^{(t)})/(|L^{(t)}| + 1)$ in the log-likelihood is less than $10^{-6}$.

{\bf Two-way ANOVA:} We simulated data from a two-way ANOVA random effects model
\begin{eqnarray*}
	y_{ijk} = \mu + \alpha_i + \beta_j + (\alpha \beta)_{ij} + \epsilon_{ijk}, 
\end{eqnarray*}
where $\alpha_i \sim N(0,\sigma_1^2)$, $\beta_j \sim N(0,\sigma_2^2)$, $(\alpha \beta)_{ij} \sim N(0, \sigma_3^2)$, and $\epsilon_{ijk} \sim N(0,\sigma_e^2)$ are jointly independent. This corresponds to $m=4$ variance components. In the simulation, we set $\sigma_2^2 = \sigma_3^2 = \sigma_e^2$ and varied the ratio $\sigma_1^2 / \sigma_e^2$; the numbers of levels $a$ and $b$ in factor 1 and factor 2, respectively; and the number of observations $c$ in each combination of factor levels. For each simulation scenario, we simulated 50 replicates. The sample size was $n= abc$ for each replicate.

Tables \ref{tab:iter5} and \ref{tab:run5} show the average number of iterations and the average runtimes when there are $a=b=5$ levels of each factor. Based on these results and further results not shown for other combinations of $a$ and $b$, we draw the following conclusions. Fisher scoring takes the fewest iterations. The MM algorithm always takes fewer iterations than the EM algorithm. Accelerated MM further improves the convergence rate of MM. The faster rate of convergence of Fisher scoring is outweighed by the extra cost of evaluating and inverting the covariance matrix. When the sample size $n=abc$ is large, Fisher scoring takes much longer than either EM or MM.

{\bf Genetic model:} We simulated a quantitative trait $\ybf$ from a genetic model with two variance components and covariance matrix $\Omegabf = \sigma_a^2 \widehat \Phibf + \sigma_e^2 \Ibf$, where $\widehat \Phibf$ is a full-rank empirical kinship matrix estimated from the genome-wide measurements of 212 individuals using Option 29 of the Mendel software \citep{Lange13Mendel,Lange05association}. In this example, Fisher scoring excels at smaller $\sigma_a^2 / \sigma_e^2$ ratios, while accelerated MM is fastest at larger $\sigma_a^2 / \sigma_e^2$ ratios.

In summary, the MM algorithm appears competitive even in small-scale examples. Modern applications often involve a large number of variance components. In this setting, the EM algorithm suffers from slow convergence and Fisher scoring from an extremely high cost per iteration. Our genomic example in Section \ref{sec:numerical-example} reinforces this point.

\begin{table}
	\centering
	\begin{tabular}{rrrrrr}
	\toprule
		$\sigma^2_{1}/\sigma^2_{e}$ & Method &
        \multicolumn{4}{c}{$c=$ \# observations per combination} \\
		\cline{3-6}
		 & & 2 & 8 & 20 & 50  \\
\midrule
		   0 & MM & 34.52(15.79) & 25.90(8.69) & 18.62(7.22) & 15.48(5.34) \\
		&aMM & 16.68(5.69) & 13.48(3.47) & 11.76(3.12) & 10.88(2.43) \\
		&EM & 123.70(63.72) & 61.58(31.36) & 38.44(18.58) & 25.66(10.31) \\
		&FS &  6.10(1.09) &  6.74(0.99) &  6.68(0.79) &  6.36(0.72) \\[0.25cm]

		 0.05 & MM & 27.78(13.05) & 22.82(8.96) & 19.82(6.55) & 15.48(3.97) \\
		&aMM & 14.80(4.33) & 12.32(3.27) & 12.08(2.62) & 11.20(2.52) \\
		&EM & 108.04(62.58) & 58.42(33.67) & 43.52(19.48) & 27.62(12.47) \\
		&FS &  6.20(1.29) &  6.72(1.25) &  6.62(0.73) &  6.60(1.07) \\[0.25cm]

		 0.1 & MM & 31.26(14.90) & 23.38(9.21) & 16.84(6.72) & 14.88(4.56) \\
         &aMM & 15.96(5.65) & 12.72(3.59) & 10.36(2.51) & 10.80(2.46) \\
		&EM & 112.12(72.70) & 62.26(28.87) & 34.86(22.61) & 24.10(11.96) \\
		&FS &  6.10(1.25) &  6.90(0.79) &  6.48(0.86) &  6.52(0.86) \\[0.25cm]

		 1 & MM & 29.72(15.85) & 22.72(10.86) & 17.78(8.18) & 13.94(4.73) \\
		&aMM & 15.24(5.60) & 12.40(4.12) & 10.72(2.70) & 10.24(2.01) \\
		&EM & 85.86(63.85) & 41.50(30.46) & 28.40(20.02) & 21.36(13.86) \\
		&FS &  5.96(1.19) &  6.90(0.91) &  6.36(1.05) &  6.44(0.93) \\[0.25cm]

		 10 & MM & 16.46(9.74) & 13.28(7.75) & 12.80(6.41) & 10.74(3.67) \\
		&aMM & 11.60(3.70) &  9.36(2.78) &  9.04(3.00) &  8.68(2.54) \\
		&EM & 24.50(32.87) & 16.18(23.06) & 15.10(16.55) & 12.36(11.13) \\
		&FS &  6.98(0.80) &  6.96(0.70) &  6.74(0.83) &  6.76(0.52) \\[0.25cm]

		 20 & MM & 17.34(10.70) & 14.20(6.79) & 11.58(4.46) & 10.16(4.26) \\
		&aMM & 12.12(5.96) &  9.92(2.68) &  8.92(2.07) &  8.48(2.00) \\
		&EM & 31.08(42.11) & 20.50(24.55) & 10.84(10.86) &  8.98(8.94) \\
		&FS &  7.18(0.98) &  7.02(0.82) &  6.90(0.74) &  6.78(0.79) \\[0.25cm]

\bottomrule
\end{tabular}
\caption{Average iterations until convergence for MM, quasi-Newton accelerated MM (aMM), EM, and Fisher scoring (FS) for fitting a two-way ANOVA model with $a = b = 5$ levels of both factors. Standard errors are given in parentheses.}
    \label{tab:iter5}
\end{table}

\begin{table}
	\centering
	\begin{tabular}{rrrrrr}
	\toprule
		$\sigma^2_{1}/\sigma^2_{e}$ & Method &
        \multicolumn{4}{c}{$c$ = \# observations per combination} \\
		\cline{3-6}
		 & & 2 & 8 & 20 & 50  \\
\midrule
		0 & MM & 30.50(81.58) & 132.50(48.51) & 739.80(272.66) & 4004.10(1317.90) \\
		&aMM& 16.76(36.57) & 92.06(34.42) & 638.76(231.86) & 3691.90(873.70) \\
		&EM & 70.76(43.58) & 376.82(184.27) & 1912.30(918.56) & 8276.98(3269.26) \\
		&FS & 17.52(27.99) & 241.06(51.79) & 4039.73(6955.91) & 43315.42(76563.64) \\[0.25cm]

		 0.05 & MM & 16.79(11.73) & 117.33(50.33) & 867.41(296.24) & 4083.60(1016.71) \\
		&aMM& 14.23(14.21) & 80.93(31.83) & 692.54(196.19) & 3841.10(948.55) \\
		&EM & 66.73(44.87) & 376.78(206.19) & 2291.69(1035.80) & 9054.02(3989.52) \\
		&FS & 13.33(18.33) & 253.10(61.72) & 3198.17(379.00) & 34057.87(7132.36) \\[0.25cm]

		 0.1 & MM & 17.01(8.97) & 122.03(55.77) & 733.37(329.08) & 3992.52(1166.67) \\
		&aMM& 12.39(8.86) & 88.34(33.87) & 593.27(174.93) & 3745.50(922.99) \\
		&EM & 76.65(53.37) & 389.45(179.41) & 1814.91(1152.64) & 7951.40(3810.34) \\
		&FS & 10.24(8.57) & 257.63(45.53) & 3140.15(481.08) & 33490.67(4533.51) \\[0.25cm]

		 1 & MM & 16.50(13.08) & 112.94(52.73) & 736.32(322.26) & 3746.87(1221.92) \\
		&aMM&  9.86(4.10) & 80.98(39.49) & 585.98(158.65) & 3536.90(754.48) \\
		&EM & 56.93(48.16) & 267.86(194.49) & 1465.45(986.31) & 7079.60(4430.10) \\
		&FS & 15.75(17.26) & 262.49(44.28) & 3003.68(481.92) & 33215.82(4801.38) \\[0.25cm]

		 10 & MM & 10.80(11.16) & 70.94(47.94) & 545.71(256.97) & 2316.96(1022.51) \\
		&aMM&  8.64(4.17) & 62.50(26.13) & 483.63(183.47) & 2317.43(1061.57) \\
		&EM & 21.51(31.61) & 113.76(158.82) & 803.36(816.05) & 3256.65(2624.29) \\
		&FS & 12.32(9.81) & 261.85(37.84) & 3190.52(394.75) & 26163.81(8451.96) \\[0.25cm]

		 20 & MM &  8.83(5.05) & 104.94(54.66) & 552.13(190.42) & 1706.71(680.84) \\
		&aMM&  9.57(9.80) & 92.94(35.84) & 524.70(137.22) & 1750.99(489.96) \\
		&EM & 23.13(31.17) & 175.12(198.18) & 642.39(576.82) & 2007.86(1901.66) \\
		&FS & 12.71(11.90) & 340.81(48.29) & 3543.18(464.36) & 18796.59(2445.74) \\[0.25cm]
\hline
	\end{tabular}
       	\caption{Average run times ($\times 10^{-3}$ seconds) of MM, quasi-Newton accelerated MM (aMM), EM, and Fisher scoring (FS) for fitting a two-way ANOVA model with $a = b = 5$ levels of both factors. Standard errors are given in parentheses.}
    \label{tab:run5}
\end{table}

\begin{table}
	\centering
	\begin{tabular}{crrrr}\\
    
    \hline
		
	$\sigma^2_a/\sigma^2_{e}$ & Method & Iteration & Runtime ($10^{-3}$ sec) & Objective   \\

\hline
		 0 & MM  & 88.10(29.01) & 778.24(305.37) & -374.35(9.82)  \\
		  & aMM  & 23.65(5.74) & 293.16(146.23) & -374.34(9.82)  \\
		  & EM  & 231.93(123.39) & 3509.02(1851.11) & -374.41(9.83)  \\
		  & FS  &  5.05(1.24) & 137.76(65.74) & -374.36(9.83) \\[0.25cm]
        
		 0.05  & MM  & 84.97(31.18) & 710.56(260.24) & -377.19(10.85)  \\
		  & aMM  & 23.05(5.45) & 272.04(67.01) & -377.18(10.85)  \\
		  & EM  & 220.57(124.70) & 3292.87(1865.91) & -377.25(10.85)  \\
		  & FS  &  5.08(1.21) & 136.47(33.18) & -377.21(10.83)\\[0.25cm]
		 0.1  & MM  & 82.45(34.39) & 673.96(268.23) & -379.62(10.54)  \\
		  & aMM  & 22.55(6.01) & 269.55(86.69) & -379.61(10.54)  \\
		  & EM  & 199.70(113.47) & 2917.71(1607.33) & -379.68(10.54)  \\
		  &FS  &  4.97(1.03) & 129.71(40.51) & -379.62(10.54)\\[0.25cm]
		 1 & MM  & 31.00(15.59) & 160.21(80.45) & -409.66(11.26)  \\
		  & aMM  & 12.55(5.38) & 90.21(43.54) & -409.66(11.26)  \\
		  & EM  & 51.10(28.70) & 550.55(321.89) & -409.67(11.26)  \\
		  & FS  &  4.60(0.59) & 80.28(25.56) & -409.66(11.26)\\[0.25cm]
         10 & MM  & 72.67(39.23) & 374.80(209.31) & -532.57(9.11)  \\
		  & aMM  & 20.15(10.18) & 146.25(81.06) & -531.24(9.28)  \\
		  & EM  & 294.20(717.05) & 3079.82(7520.30) & -532.71(9.11)  \\
		  & FS  & 10.18(4.92) & 168.63(80.34) & -532.08(9.21)  \\[0.25cm]
		20 & MM  & 78.35(34.32) & 425.40(188.08) & -591.36(7.05)  \\
		  & aMM  & 14.80(6.53) & 117.14(71.14) & -589.13(7.15)  \\
		  & EM  & 362.07(764.60) & 4144.92(8862.65) & -591.62(6.82)  \\
		  & FS  & 10.93(4.75) & 181.48(83.96) & -590.68(7.08)\\[0.25cm]
\hline
	\end{tabular}
	\caption{Average performance of MM, quasi-Newton accelerated MM (aMM), EM, and Fisher scoring (FS) for fitting a genetic model. Standard errors are given in parentheses.}
    \label{tab:mendel}
\end{table}

\section{Global convergence of the MM algorithm}

The KKT necessary conditions for a local maximum 
$\sigmabf^2=(\sigma_1^2, \ldots, \sigma_m^2)$ of the log-likelihood 
\eqref{eqn:vc-loglike} require each component of the score vector to satisfy  
\begin{eqnarray*}
	\frac{\partial}{\partial \sigma_i^2} L(\sigmabf^2) &  \in &  \begin{cases}
	 \{0\} & \sigma_i^2>0 \\
	(-\infty,0] & \sigma_i^2=0 \, .
	\end{cases}
\end{eqnarray*}
In this section we establish the global convergence of Algorithm \ref{algo:varcomp-univariate-y} to a KKT point. To reduce the notational burden, we assume that $\Xbf$ is null and omit estimation of fixed effects $\betabf$. The analysis easily extends to the MLE case. Our convergence analysis relies on characterizing the properties of the objective function $L(\sigmabf^2)$ and the MM algorithmic mapping $\sigmabf^2 \mapsto M(\sigmabf^2)$ defined by equation \eqref{eqn:MM-update-univariate}. Special attention must be paid to the boundary values $\sigma_i^2 = 0$. We prove convergences for two cases, which cover most applications.
\begin{assumption}
All $\Vbf_i$ are positive definite.
\end{assumption}
\begin{assumption}
$\Vbf_1$ is positive definite, each $\Vbf_i$ is nontrivial, ${\cal H} = \text{span} \{\Vbf_2, \ldots, \Vbf_m\}$ has dimension $q < n$, and $\ybf \notin {\cal H}$.
\end{assumption}
The genetic model in Section \ref{sec:univ-y} satisfies Assumption 1, while the two-way ANOVA model satisfies Assumption 2. The key condition $\ybf \notin \text{span}\{\Vbf_2, \ldots, \Vbf_m\}$ in the second case is critical for the existence of an MLE or an REML \citep{DemidenkoMassam99MVCMLEExist,GrzadzielMichalski14VCMLEExist}. We will derive a sequence of lemmas en route to the global convergence result declared in Theorem \ref{thm:convergence}. 

\begin{lemma}
\label{lemma:coerciveness}
Under Assumption 1 or 2, the log-likelihood function \eqref{eqn:vc-loglike} is coercive in the sense that the 
super-level set $S_c =\{\sigmabf^2 \ge \zerobf : L(\sigmabf^2) \ge c\}$ is compact for every $c$.
\end{lemma}
\begin{proof}
Let us first prove the assertion when all of the covariance matrices $\Vbf_i$ 
are positive definite. If we set $r = \|\sigmabf^2\|_1$ and $\alpha_i = r^{-1}\sigma_i^2$
for each $i$, then the log-likelihood satisfies
\begin{eqnarray*}
L(\sigmabf^2) & = & -\frac{n}{2} \ln r 
- \frac{1}{2} \ln \det \Big( \sum_{i=1}^m \alpha_i \Vbf_i \Big) 
-\frac{1}{2r}\ybf^T \Big(\sum_{i=1}^m \alpha_i \Vbf_i \Big)^{-1} \ybf.
\end{eqnarray*}
The functions $\ln \det \Big(\sum_{i=1}^m \alpha_i \Vbf_i \Big)$ and
$\ybf^T \Big(\sum_{i=1}^m \alpha_i \Vbf_i \Big)^{-1} \ybf$ of $\alphabf$ are
defined and continuous on the unit simplex and hence bounded there.
The dominant term $-\frac{n}{2} \ln r$ of the loglikelihood
tends to $-\infty$ as $r$ tends to $\infty$.

To prove the assertion under Assumption 2, consider first the case $\Vbf_1 = \Ibf_n$. Setting $\alpha_i = \sigma_i^2 / \sigma_1^2$ for $i=2,\ldots, m$ reduces the loglikelihood to
\begin{eqnarray}
	L(\sigma_1^2, \alphabf) & = & - \frac{n}{2} \ln \sigma_1^2 - \frac{1}{2} \ln \det \Big(\Ibf_n + \sum_{i=2}^m \alpha_i \Vbf_i \Big) - \frac{1}{2 \sigma_1^2} \ybf^T \Big( \Ibf_n + \sum_{i=2}^m \alpha_i \Vbf_i \Big)^{-1}\!\! \ybf. \label{eqn:alpha-param}
\end{eqnarray}
The middle term on the right satisfies 
\begin{eqnarray*}
- \frac{1}{2} \ln \det \Big(\Ibf_n + \sum_{i=2}^m \alpha_i \Vbf_i \Big) & \le & 0
\end{eqnarray*}
because $\det \,(\Ibf_n + \sum_{i=2}^m \alpha_i \Vbf_i ) \,\ge \, \det \Ibf_n =  1.$
Now let $\Ubf = (\Ubf_q, \Ubf_{n-q})$ be an $n \times n$ orthogonal matrix whose left columns 
$\Ubf_q$ span ${\cal H}$ and whose right columns $\Ubf_{n-q}$ span ${\cal H}^\perp$. The identity
\begin{eqnarray*}
	\Ubf^T \Big( \Ibf_n +  \sum_{i=2}^m \alpha_i \Vbf_i \Big) \Ubf &=& \begin{pmatrix}
	\Ibf_q + \sum_{i=2}^m \alpha_i \Ubf_q^T \Vbf_i \Ubf_q & \zerobf \\
	\zerobf & \Ibf_{n-q} \end{pmatrix}
\end{eqnarray*}
follows from the orthogonality relations $\Ubf_{n-q}^T \Vbf_i  \, = \,\Ubf_{n-q}^T \Ubf_q \,= \, \zerobf_{(n-q) \times n}$. This in turn implies  
\begin{eqnarray*}
 \Big( \Ibf_n +  \sum_{i=2}^m \alpha_i \Vbf_i \Big)^{-1} & = &  \Ubf \begin{pmatrix}
(\Ibf_q + \sum_{i=2}^m \alpha_i \Ubf_q^T \Vbf_i \Ubf_q)^{-1} & \zerobf \\
\zerobf & \Ibf_{n-q} \end{pmatrix} \Ubf^T \\
&\succeq & \Ubf \begin{pmatrix}
\zerobf & \zerobf \\
\zerobf & \Ibf_{n-q} \end{pmatrix} \Ubf^T \\
	& = & \Ubf_{n-q} \Ubf_{n-q}^T.
\end{eqnarray*}
Therefore the quadratic term in equation \eqref{eqn:alpha-param} is bounded below by the positive constant
\begin{eqnarray*}
\ybf^T \Big( \Ibf_n + \sum_{i=2}^m \alpha_i \Vbf_i \Big)^{-1} \!\! \ybf 
& \ge & \ybf^T \Ubf_{n-q} \Ubf_{n-q}^T \ybf \amp = \amp 
\|\mathbb{P}_{{\cal H}^\perp} \ybf \|^2 \amp > \amp 0.
\end{eqnarray*}
Here the assumption $\ybf \notin {\cal H}$ guarantees the projection property
$\mathbb{P}_{{\cal H}^\perp} \ybf  \ne \zerobf$.

Next we show that the loglikelihood tends to $-\infty$ when $\sigma_1^2$ tends to $0$ or $\infty$ or 
when $\|\alphabf\|_2$ tends to $\infty$.  The second of the two inequalities 
\begin{eqnarray*}
L(\sigma_0^2, \alphabf) & \le & - \frac{n}{2} \ln \sigma_1^2 - 
 \frac{1}{2} \ln \det \Big(\Ibf_n 
+ \sum_{i=2}^m \alpha_i \Vbf_i \Big) 
-\frac{1}{2 \sigma_1^2}\|\mathbb{P}_{{\cal H}^\perp} \ybf\|^2 \\
& \le & - \frac{n}{2} \ln \sigma_1^2 - 
\frac{1}{2 \sigma_1^2}\|\mathbb{P}_{{\cal H}^\perp} \ybf\|^2
\end{eqnarray*}
renders the claim about $\sigma_1^2$ obvious. To prove the claim about $\alphabf$, we 
make the worst case choice $\sigma_i^2 = \|\mathbb{P}_{{\cal H}^\perp} \ybf\|^2$
in the first inequality.  
It follows that  
\begin{eqnarray*}
L(\sigma_0^2, \alphabf) & \le & - \frac{1}{2} \ln \det \Big(\Ibf_n 
+ \sum_{i=2}^m \alpha_i \Vbf_i \Big)
- \frac{n}{2} \ln \|\mathbb{P}_{{\cal H}^\perp} \ybf\|^2 -\frac{n}{2}.
\end{eqnarray*}
If $\alpha_j$ tends to $\infty$, then the inequality
\begin{eqnarray*}
- \frac{1}{2} \ln \det \Big(\Ibf_n + \sum_{i=2}^m \alpha_i \Vbf_i \Big) & \le &
- \frac{1}{2} \ln \det \Big(\Ibf_n + \alpha_j \Vbf_j \Big) \amp = \amp
- \frac{1}{2} \sum_{k=1}^n \ln (1 + \alpha_j \lambda_{jk})  
\end{eqnarray*}
holds, where the $\lambda_{jk}$ are the eigenvalues of $\Vbf_j$.
At least one of these eigenvalues is positive because $\Vbf_j$ is nontrivial.
It follows that $L(\sigma_0^2, \alphabf)$ tends to $-\infty$ 
in this case as well.

For the general case where $\Vbf_1$ is non-singular but not necessarily $\Ibf_n$, let 
$\Vbf^{1/2}_1$ be the symmetric square root of $\Vbf_1$ and write
\begin{eqnarray*}
	\Vbf_1 + \sum_{i=2}^m \sigma_i^2 \Vbf_i = \Vbf^{1/2}_1 \left( \Ibf + \sum_{i=2}^m \sigma_i^2 \Vbf_1^{-1/2} \Vbf_i \Vbf_1^{-1/2} \right) \Vbf^{1/2}_1.
\end{eqnarray*}
The above arguments still apply since each $\Vbf_1^{-1/2} \Vbf_i \Vbf_1^{-1/2}$ is nontrivial
and $\ybf$ belongs to the $\text{span} \{\Vbf_2, \ldots, \Vbf_m\} = \Sbf$ if and only if $\Vbf_1^{-1/2}\ybf$ belongs to $\Vbf^{-1/2}_1\Sbf \Vbf_1^{-1/2}$.
\end{proof}

\begin{lemma}
\label{lemma:monotonicity-stationarity}
The iterates possess the ascent property $L(M(\sigmabf^{2(t)})) \ge L(\sigmabf^{2(t)})$. Furthermore,
when $L(M(\sigmabf_*^2)) = L(\sigmabf_*^2)$, $\sigmabf_*^2$ fulfills the
fixed point condition $M(\sigmabf_*^2) = \sigmabf_*^2$, and each component satisfies either (i) $\sigma_{*i}^2 = 0$ or (ii) $\sigma_{*i}^2 > 0$ and $\frac{\partial}{\partial \sigma_i^2} L(\sigmabf_*^2) = 0$.
\end{lemma}
\begin{proof}
The ascent property is built into any MM algorithm. Suppose $L(M(\sigmabf_*^2)) = L(\sigmabf_*^2)$ at a point $\sigmabf_*^2 \in \mathbb{R}_+^m$. Then equality must hold in the string of inequalities
\eqref{eqn:MM-monotonicity}. It follows that 
\begin{eqnarray*}
g(M(\sigmabf_*^2) \mid \sigmabf_*^2) & = & g(\sigmabf_*^2 \mid \sigmabf_*^2)
\end{eqnarray*}
and hence that $M(\sigmabf_*^2) = \sigmabf_*^2$. If $\sigma_{*i}^2 > 0$,
the stationarity condition
\begin{eqnarray*}
\frac{\partial}{\partial \sigma_i^2} L(\sigmabf_*^2) & = &
\frac{\partial}{\partial \sigma_i^2} g(\sigmabf_*^2 \mid \sigmabf_*^2) \amp = \amp 0 
\end{eqnarray*}
applies. The equivalence of the two displayed partial derivatives is a consequence of
the fact that the difference $f(\sigmabf^2)-g(\sigmabf^2 \mid \sigmabf_*^2)$ achieves
its minimum of 0 at $\sigmabf^2 = \sigmabf_*^2$.
\end{proof}

\begin{lemma}
The distance between successive iterates $\|\sigmabf^{2(t+1)} - \sigmabf^{2(t)}\|_2$ converges to 0.
\label{lemma:successive_iterates}
\end{lemma}
\begin{proof} Suppose on the contrary that $\|\sigmabf^{2(t+1)} - \sigmabf^{2(t)}\|_2$ does not converge to 0. Then one can extract a subsequence $\{t_k\}_{k \ge 1}$ such that 
\begin{eqnarray}
\|\sigmabf^{2(t_k+1)} - \sigmabf^{2(t_k)}\|_2 \ge \epsilon > 0 \label{eqn:successive-diff}
\end{eqnarray}
for all $k$. Let $C_0$ be the compact super-level set 
$\{\sigmabf^2: L(\sigmabf^2) \ge L(\sigmabf^{2(0)})\}$.
Since the sequence $\{\sigmabf^{2(t_k)}\}_{k \ge 1}$ is confined to $C_0$, one can pass to a subsequence if necessary and assume that $\sigmabf^{2(t_k)}$ converges to a limit 
$\sigmabf_{*}^2$ and that $\sigmabf^{2(t_k+1)}$ converges to a limit $\sigmabf_{**}^2$.
Taking limits in the relation $\sigmabf^{2(t_k+1)}=M(\sigmabf^{2(t_k)})$ and invoking the continuity $M(\sigmabf^2)$ imply that $\sigmabf_{**}^2=M(\sigmabf_{*}^2)$. 
Because the sequence $L(\sigmabf^{2(t_k)})$ is monotonically increasing in $k$ and bounded above 
on $C_0$, it converges to a limit $L_*$. Hence, the continuity of $L(\sigmabf^2)$ implies
\begin{eqnarray*}
L(\sigmabf_*^2) & = & \amp \lim_k L(\sigmabf^{2(t_k)}) 
\amp = \amp L_* \amp = \amp \lim_k L(\sigmabf^{2(t_k+1)}) \amp = \amp L(\sigmabf_{**}^2)
\amp = \amp L(M(\sigmabf_*^2)). 
\end{eqnarray*}
Lemma \ref{lemma:monotonicity-stationarity} therefore gives 
$\sigmabf_{**}^2 = M(\sigmabf_{*}^2) = \sigmabf_*^2$, contradicting
the bound $\|\sigmabf_*^2 - \sigmabf_{**}^2\|_2 \ge \epsilon$ entailed by 
inequality \eqref{eqn:successive-diff}.
\end{proof}

\begin{theorem}
\label{thm:convergence}
The MM sequence $\{\sigmabf^{2(t)}\}_{t \ge 0}$ has at least one limit point.  
Every limit point is a fixed point of $M(\sigmabf^2)$. If the set of fixed points 
is discrete, then the MM sequence converges to one of them. Finally, when the
iterates converge, their limit is a KKT point. 
\end{theorem}
\begin{proof}
The sequence $\{\sigmabf^{2(t)}\}_{t \ge 0}$ is contained in the super-level
compact set $C_0$ defined in Lemma \ref{lemma:successive_iterates} and therefore admits a convergent 
subsequence $\sigmabf^{2(t_k)}$ with limit $\sigmabf^{2(\infty)}$. As argued in Lemma \ref{lemma:successive_iterates}, $L(\sigmabf^{2(\infty)}) = L(M(\sigmabf^{2(\infty)}))$. Lemma
\ref{lemma:monotonicity-stationarity} now implies that $\sigmabf^{2(\infty)}$ is a fixed point of 
the algorithm map $M(\sigmabf^2)$.

According to Ostrowski's theorem \citep[Proposition 8.2.1]{Lange10NumAnalBook}, the set of limit points of a bounded sequence $\{\sigmabf^{2(t)}\}_{t \ge 0}$ is connected and compact provided $\|\sigmabf^{2(t+1)} - \sigmabf^{2(t)}\|_2 \to 0$. If the set of fixed points is discrete,
then the connected subset of limit points reduces to a single point. Hence, the bounded
sequence $\sigmabf^{2(t)}$ converges to this point. When the limit exists, one can check
that $\sigmabf^{2(\infty)}$ satisfies the KKT conditions by proving that each zero component 
of $\sigmabf^{2(\infty)}$ has a non-positive partial derivative. Suppose on the contrary $\sigma_i^{2(\infty)}=0$ and $\frac{\partial}{\partial \sigma_i^2} L(\sigmabf^{2(\infty)}) > 0$. By continuity $\frac{\partial}{\partial \sigma_i^2} L(\sigmabf^{2(t)})  > 0$ for all large $t$. Therefore, $\sigma_i^{2(t+1)} > \sigma_i^{2(t)}$ for all large $t$ by the observation made after equation \eqref{eqn:L-gradient}. This behavior is inconsistent with the assumption that $\sigma_i^{2(t)} \to 0$.
\end{proof}

\section{MM versus EM}

Examination of Tables \ref{tab:run5} and \ref{tab:mendel} suggests that the MM algorithm usually converges faster than the EM algorithm. We now provide theoretical justification for this observation. Again for notational convenience, we consider the REML case where $\Xbf$ is null.
Since the EM principle is just a special instance of the MM principle, we can compare their convergence properties in a unified framework. Consider an MM map $M(\thetabf)$ for maximizing the objective function $f(\thetabf)$ via the surrogate function $g(\thetabf \mid \thetabf^{(t)})$. Close to the optimal point 
$\thetabf^{\infty}$,
\begin{eqnarray*}
    \thetabf^{(t+1)} - \thetabf^{\infty} & \approx & dM(\thetabf^{\infty}) (\thetabf^{(t)} - \thetabf^{\infty}),
\end{eqnarray*}
where $dM(\thetabf^{\infty})$ is the differential of the mapping $M$ at the optimal point $\thetabf^{\infty}$ of $f(\thetabf)$. Hence, the local convergence rate of the sequence $\thetabf^{(t+1)} = M(\thetabf^{(t)})$ coincides with the spectral radius of $dM(\thetabf^{\infty})$. Familiar calculations \citep{McLachlan08EMBook,Lange10NumAnalBook} demonstrate that
\begin{eqnarray*}
    dM(\thetabf^\infty) & = & \Ibf - [d^{2}g(\thetabf^\infty \mid \thetabf^\infty)]^{-1} d^2 f(\thetabf^\infty).  \label{eqn:MM-map-diff}
\end{eqnarray*}
In other words, the local convergence rate is determined by how well the surrogate surface 
$g(\thetabf \mid \thetabf^{\infty})$ approximates the objective surface $f(\thetabf)$ 
near the optimal point $\thetabf^\infty$. In the EM literature, $dM(\thetabf^\infty)$ is called the \emph{rate matrix} \citep{MengRubin91SEM}. Fast convergence occurs when the surrogate $g(\thetabf \mid \thetabf^{\infty})$ hugs the objective $f(\thetabf)$ tightly around $\thetabf^{\infty}$. Figure \ref{fig:MM-vs-EM} shows a case where the MM surrogate locally dominates the EM surrogate. We demonstrate that this is no accident.

\citet{McLachlan08EMBook} derive the EM surrogate 
\begin{eqnarray*}
g_{\text{EM}}(\sigmabf^{2} \!\mid \! \sigmabf^{2(t)})
&=& -\frac{1}{2} \sum_{i=1}^{m}\Big[\rank(\Vbf_i) \ln \sigma_{i}^2 + \rank(\Vbf_i) \frac{\sigma_{i}^{2(t)}}{\sigma_{i}^2}
- \frac{\sigma_{i}^{4(t)}}{\sigma_{i}^2}\tr(\Omegabf^{-(t)}\Vbf_{i})\Big] \\
&  & - \frac{1}{2} \sum_{i=1}^m \frac{\sigma_{i}^{4(t)}}{\sigma^2_{i}}\ybf ^T \Omegabf^{-(t)} \Vbf_i \Omegabf^{-(t)} \ybf
\end{eqnarray*}
minorizing the log-likelihood up to an irrelevant constant. Section S.3 
of the Supplementary Materials gives a detailed derivation for the more general multivariate response case. The rank of the covariance matrix $\Vbf_i$ appears because $\Vbf_i$ may not be invertible. Both of the surrogates $g_{\text{EM}}(\sigmabf^{2} \!\mid \!\sigmabf^{2(\infty)})$ and $g_{\text{MM}}(\sigmabf^{2} \!\mid \!\sigmabf^{2(\infty)})$ are parameter separated. This implies that both second differentials 
$d^{2}g_{\text{EM}}(\sigmabf^{2(\infty)} \!\mid \!\sigmabf^{2(\infty)})$ and $d^{2}g_{\text{MM}}(\sigmabf^{2(\infty)} \!\mid \!\sigmabf^{2(\infty)})$ are diagonal. A small diagonal entry of either matrix indicates fast convergence of the corresponding variance component. Our next result shows that, under Assumption 1, on average the diagonal entries of $d^{2}g_{\text{EM}}(\sigmabf^{2(\infty)} \!\mid \!\sigmabf^{2(\infty)})$ dominate those of $d^{2}g_{\text{MM}}(\sigmabf^{2(\infty)} \!\mid \!\sigmabf^{2(\infty)})$ when $m>2$. Thus, the EM algorithm tends to converge more slowly than the MM algorithm, and the difference is more pronounced as the number of variance components $m$ grows. 
\begin{theorem}
\label{thm:local-conv-rate}
Let $\sigmabf^{2(\infty)} \succ \zerobf_m$ be a common limit point of the EM and MM algorithms. Then 
both second differentials $d^{2}g_{\text{MM}}(\sigmabf^{2(\infty)}\!\mid \!\sigmabf^{2(\infty)})$ and $d^{2}g_{\text{EM}}(\sigmabf^{2(\infty)}\! \mid \!\sigmabf^{2(\infty)})$ are diagonal with 
\begin{eqnarray*}
d^{2}g_{\text{EM}}(\sigmabf^{2(\infty)} \mid  \sigmabf^{2(\infty)})_{ii}
& = & -\frac{\rank(\Vbf_i)}{2 \sigma_{i}^{4(\infty)}} \\
d^{2}g_{\text{MM}}(\sigmabf^{2(\infty)} \mid  \sigmabf^{2(\infty)})_{ii}
& = & -\frac{\ybf^T \Omegabf^{-(\infty)} \Vbf_i \Omegabf^{-(\infty)} \ybf}{\sigma_i^{2(\infty)}} 
\amp = \amp -\frac{\tr(\Omegabf^{-(\infty)}\Vbf_i)}{\sigma_i^{2(\infty)}}.
\end{eqnarray*}
Furthermore, the average ratio
\begin{eqnarray*}
\frac{1}{m} \sum_{i=1}^m \frac{d^{2}g_{\text{MM}}(\sigmabf^{2(\infty)} \mid \sigmabf^{2(\infty)})_{ii}}{d^{2}g_{\text{EM}}(\sigmabf^{2(\infty)} \mid  \sigmabf^{2(\infty)})_{ii}} & = &  \frac{2}{mn} \sum_{i=1}^m \tr(\Omegabf^{-(\infty)}\sigma_i^{2(\infty)}\Vbf_i) 
\amp = \amp \frac{2}{m} \amp < \amp 1
\end{eqnarray*}
for $m>2$ when all $\Vbf_i$ have full rank $n$.
\end{theorem}
\begin{proof} See Section \ref{sec:local-conv-rate} of the Supplementary Materials.
\end{proof}
Both the EM and MM algorithms must evaluate the traces $\tr(\Omegabf^{-(t)} \Vbf_i)$ and 
quadratic forms $(\ybf - \Xbf \betabf^{(t)})^T \Omegabf^{-(t)} \Vbf_i \Omegabf^{-(t)} (\ybf - \Xbf \betabf^{(t)})$ at each iteration. Since these quantities are also the building blocks of the approximate rate matrices $d^{2}g(\sigmabf^{2(t)} \!\mid \!\sigmabf^{2(t)})$, one can rationally choose either the EM or MM updates based on which has smaller diagonal entries measured by the $\ell_1$, $\ell_2$, or $\ell_\infty$ norms. At negligible extra cost, this produces a hybrid algorithm that retains the ascent property and enjoys the better of the two convergence rates.

\section{Extensions}

Besides its competitive numerical performance, Algorithm \ref{algo:varcomp-univariate-y} is attractive for its simplicity and ease of generalization. In this section, we outline MM algorithms for multivariate response models possibly with missing data, linear mixed models, MAP estimation, penalized estimation, and generalized estimating equations.

\subsection{Multivariate response model}

Consider the multivariate response model with $n \times d$ response matrix $\Ybf$, mean 
$\E \, \Ybf = \Xbf \Bbf$, and covariance
\begin{eqnarray*}
	\Omegabf &=& \Cov(\vect \Ybf) \amp = \amp  \sum_{i=1}^m \Gammabf_i \otimes \Vbf_i.
\end{eqnarray*}
The $p \times d$ coefficient matrix $\Bbf$ collects the fixed effects, the $\Gammabf_i$ are 
unknown $d \times d$ covariance matrices, and the $\Vbf_i$ are known $n \times n$ covariance
matrices. If the vector $\vect \Ybf$ is normally distributed, then $\Ybf$ equals a sum
of independent matrix normal distributions \citep{Gupta1999matrix}. We now make this assumption and pursue estimation of $\Bbf$ and the $\Gammabf_i$, which we collectively denote as $\Gammabf$. Under
the normality assumption, the Kronecker product identity $\vect(\Cbf \Dbf \Ebf) = (\Ebf^T \otimes \Cbf )\vect(\Dbf)$ yields the log-likelihood 
\begin{eqnarray}
L(\Bbf, \Gammabf) & = & - \frac{1}{2} \ln \det \Omegabf - \frac{1}{2} \vect (\Ybf - \Xbf \Bbf)^T 
\Omegabf^{-1} \vect (\Ybf - \Xbf \Bbf) \label{eqn:loglik-mvt} \\
& = &  - \frac{1}{2} \ln \det \Omegabf - \frac{1}{2}
[\vect \Ybf - (\Ibf_d \otimes \Xbf) \vect \Bbf]^T \Omegabf^{-1} [\vect \Ybf - (\Ibf_d \otimes \Xbf) \vect \Bbf] . \nonumber  
\end{eqnarray}
Updating $\Bbf$ given $\Gammabf^{(t)}$ is accomplished by solving the general 
least squares problem met earlier in the univariate case. Maximization of the log-likelihood
\eqref{eqn:loglik-mvt} is difficult due to the requirement that each $\Gammabf_i$ be positive 
semidefinite. Typical solutions involve reparameterization of the covariance matrix \citep{PinheiroBates96CovParam}. The MM algorithm derived in this section gracefully accommodates 
the covariance constraints.

Updating  $\Gammabf$ given $\Bbf^{(t)}$ requires generalizing the minorization (\ref{eqn:matrix-ineq}).
In view of Lemma \ref{convexity_lemma} and the identities 
$(\Abf \otimes \Bbf)(\Cbf \otimes \Dbf) = (\Abf \Cbf) \otimes (\Bbf \Dbf)$
and $(\Abf \otimes \Bbf)^{-1} = \Abf^{-1} \otimes \Bbf^{-1}$, we have
\begin{eqnarray*}
\Omegabf^{(t)} \Omegabf^{-1} \Omegabf^{(t)} & = &
m \left[\frac{1}{m}\sum_{i=1}^m \Gammabf_i^{(t)}\otimes \Vbf_i \right]
\left[\frac{1}{m}\sum_{i=1}^m \Gammabf_i \otimes \Vbf_i \right]^{-1}
\left[\frac{1}{m}\sum_{i=1}^m \Gammabf_i^{(t)}\otimes \Vbf_i \right] \\
& \preceq & m  \frac{1}{m}\sum_{i=1}^m (\Gammabf_i^{(t)}\otimes \Vbf_i ) 
(\Gammabf_i \otimes \Vbf_i)^{-1}(\Gammabf_i^{(t)} \otimes \Vbf_i )  \\
& = & \sum_{i=1}^m (\Gammabf_i^{(t)} \Gammabf_i^{-1} \Gammabf_i^{(t)}) \otimes \Vbf_i ,
\end{eqnarray*}
or equivalently
\begin{eqnarray}
\Omegabf^{-1}  & \preceq & \Omegabf^{-(t)} \Big[\sum_{i=1}^m (\Gammabf_i^{(t)} \Gammabf_i^{-1} \Gammabf_i^{(t)}) \otimes \Vbf_i \Big]\Omegabf^{-(t)}.\label{matrix_normal_majorization}
\end{eqnarray}
This derivation relies on the invertibility of the matrices $\Vbf_i$.
One can relax this assumption by substituting $\Vbf_{\epsilon,i} = \Vbf_i +\epsilon \Ibf_n$
for $\Vbf_i$ and sending $\epsilon$ to 0.

The majorization \eqref{matrix_normal_majorization} and the minorization \eqref{eqn:logdet-majorization} jointly yield the surrogate
\begin{eqnarray*}
g(\Gammabf \mid \Gammabf^{(t)}) &\!\! = \!\!& -\frac{1}{2}\sum_{i=1}^m 
\Big\{ \tr [\Omegabf^{-(t)} (\Gammabf_i \otimes \Vbf_i)] \!
+ \!(\vect \,\Rbf^{(t)})^T [(\Gammabf_i^{(t)} \Gammabf_i^{-1} \Gammabf_i^{(t)}) \!\otimes \! \Vbf_i] \, (\vect \,\Rbf^{(t)}) \Big\} \!+ \!c^{(t)},
\end{eqnarray*}
where $\Rbf^{(t)}$ is the $n \times d$ matrix satisfying
$\vect \, \Rbf^{(t)} = \Omegabf^{-(t)} \vect (\Ybf - \Xbf \Bbf^{(t)})$ and
$c^{(t)}$ is an irrelevant constant. Based on the Kronecker identities
$(\vect \, \Abf)^T \vect \, \Bbf= \tr(\Abf^T \Bbf)$ and 
$\vect(\Cbf \Dbf \Ebf) = (\Ebf^T \otimes \Cbf )\vect(\Dbf)$, the surrogate can be
rewritten as 
\begin{eqnarray*}
g(\Gammabf \mid \Gammabf^{(t)}) 
& = & -\frac{1}{2}\sum_{i=1}^m \Big\{ \tr [\Omegabf^{-(t)} (\Gammabf_i \otimes \Vbf_i)]
+ \tr(\Rbf^{(t)T} \Vbf_i \Rbf^{(t)} \Gammabf_i^{(t)} \Gammabf_i^{-1} \Gammabf_i^{(t)}) \Big\}+c^{(t)} \\
& = & -\frac{1}{2}\sum_{i=1}^m \Big\{ \tr [\Omegabf^{-(t)} (\Gammabf_i \otimes \Vbf_i)]
+ \tr( \Gammabf_i^{(t)} \Rbf^{(t)T} \Vbf_i \Rbf^{(t)} \Gammabf_i^{(t)} \Gammabf_i^{-1}) \Big\}+c^{(t)} .
\end{eqnarray*}
The first trace is linear in $\Gammabf_i$ with the coefficient of entry $(\Gammabf_i)_{jk}$ equal to
\begin{eqnarray*}
\tr(\Omegabf_{jk}^{-(t)} \Vbf_i) & = & \onebf_n^T (\Vbf_i \odot \Omegabf^{-(t)}_{jk}) \onebf_n,
\end{eqnarray*}
where $\Omegabf_{jk}^{-(t)}$ is the $(j, k)$-th $n \times n$ block of $\Omegabf^{-(t)}$. 
The matrix $\Mbf_i$ of these coefficients can be written as
\begin{eqnarray*}
\Mbf_i & = & \begin{pmatrix} \onebf^T & \zerobf^T & \ldots & \zerobf^T \\
\zerobf^T & \onebf^T & \ldots & \zerobf^T \\
\vdots & \vdots & \ddots & \vdots \\
\zerobf^T & \zerobf^T & \ldots & \onebf^T \end{pmatrix}
\left[ \begin{pmatrix} \Vbf_i & \ldots & \Vbf_i \\
\vdots & \ddots & \vdots \\
\Vbf_i & \ldots & \Vbf_i \end{pmatrix} \odot \Omegabf^{(-t)} \right]
\begin{pmatrix} \onebf & \zerobf & \ldots & \zerobf \\
\zerobf & \onebf & \ldots & \zerobf \\
\vdots & \vdots & \ddots & \vdots \\
\zerobf & \zerobf & \ldots & \onebf \end{pmatrix} \\
& = & (\Ibf_d \otimes \onebf_n)^T [(\onebf_d \onebf_d^T \otimes \Vbf_i) \odot \Omegabf^{-(t)}]  
(\Ibf_d \otimes \onebf_n) .
\end{eqnarray*}  

The directional derivative of $g(\Gammabf \mid \Gammabf^{(t)})$ with respect to $\Gammabf_i$ in the direction $\Deltabf_i$ is
\begin{eqnarray*}
&  & - \frac{1}{2}\tr(\Mbf_i \Deltabf_i)+\frac{1}{2}\tr( \Gammabf_i^{(t)} \Rbf^{(t)T} \Vbf_i \Rbf^{(t)} \Gammabf_i^{(t)} \Gammabf_i^{-1}\Deltabf_i \Gammabf_i^{-1}) \\
& = & - \frac{1}{2}\tr(\Mbf_i \Deltabf_i)+\frac{1}{2}\tr(\Gammabf_i^{-1} \Gammabf_i^{(t)} \Rbf^{(t)T} \Vbf_i \Rbf^{(t)} \Gammabf_i^{(t)} \Gammabf_i^{-1}\Deltabf_i ).
\end{eqnarray*} 
Because all directional derivatives of $g(\Gammabf \mid \Gammabf^{(t)})$ vanish
at a stationarity point, the matrix equation
\begin{eqnarray}
\Mbf_i  & = & \Gammabf_i^{-1} \Gammabf_i^{(t)} \Rbf^{(t)T} \Vbf_i \Rbf^{(t)} \Gammabf_i^{(t)} \Gammabf_i^{-1} \label{eqn:mvt-stationarity} 
\end{eqnarray}
holds. Fortunately, this equation admits an explicit solution.
For positive scalers $a$ and $b$, the solution to the equation 
$b = x^{-1}a x^{-1}$ is $x = \pm \sqrt{a/b}$.
The matrix analogue of this equation is the  Riccati equation $\Bbf = \Xbf^{-1} \Abf \Xbf^{-1}$,
whose solution is summarized in the next lemma.

\begin{algorithm}[tb!]
\SetKwInOut{Input}{input}\SetKwInOut{Output}{output}
\Input{$\Ybf$, $\Xbf$, $\Vbf_1, \ldots, \Vbf_m$}
\Output{MLE $\widehat \Bbf$, $\widehat \Gammabf_1, \ldots, \widehat \Gammabf_m$}
Initialize $\Gammabf_i^{(0)}$ positive definite, $i=1,\ldots,m$ \;
\Repeat{objective value converges}{
$\Omegabf^{(t)} \gets \sum_{i=1}^m \Gammabf_i^{(t)} \otimes \Vbf_i$ \;
$\Bbf^{(t)} \gets \arg \min_{\Bbf} \, [\vect \Ybf - (\Ibf_d \otimes \Xbf) \vect \Bbf]^T \Omegabf^{-(t)} [\vect \Ybf - (\Ibf_d \otimes \Xbf) \vect \Bbf]$ \label{eqn:update-beta} \;
$\Rbf^{(t)} \gets \text{reshape}(\Omegabf^{-(t)} \vect (\Ybf - \Xbf \Bbf^{(t)}), n, d)$ \;
\For{$i=1, \ldots, m$}{
Cholesky $\Lbf_i^{(t)} \Lbf_i^{(t)T} \gets (\Ibf_d \otimes \onebf_n)^T [(\onebf_d \onebf_d^T \otimes \Vbf_i) \odot \Omegabf^{-(t)}]  (\Ibf_d \otimes \onebf_n)$ \;
$\Gammabf_i^{(t+1)} \gets \Lbf_i^{-(t)T} [\Lbf_i^{(t)T} (\Gammabf_i^{(t)} \Rbf^{(t)T} \Vbf_i \Rbf^{(t)} \Gammabf_i^{(t)}) \Lbf_i^{(t)}]^{1/2} \Lbf_i^{-(t)}$
}
}
\caption{The MM algorithm for MLE of the multivariate response model \eqref{eqn:loglik-mvt}.}
\label{algo:varcomp-mvt-y}
\end{algorithm}
 
\begin{lemma} \label{lemma:ricatti-equation}
Assume $\Abf$ and $\Bbf$ are positive definite and $\Lbf$ is the Cholesky factor of $\Bbf$. Then 
$\Ybf = \Lbf^{-T} (\Lbf^{T} \Abf \Lbf)^{1/2} \Lbf^{-1}$ is the unique positive definite solution to the matrix equation $\Bbf = \Xbf^{-1} \Abf \Xbf^{-1}$.
\end{lemma}
\begin{proof}
Direct substitution shows that $\Ybf$ solves the equivalent equation
$\Xbf \Bbf \Xbf = \Abf$. To show uniqueness, suppose $\Ybf^{-1} \Abf \Ybf^{-1} = \Bbf$ and $\Zbf^{-1} \Abf \Zbf^{-1} = \Bbf$. The equations 
\begin{eqnarray*}
(\Bbf^{1/2} \Ybf \Bbf^{1/2})^2 & = & \Bbf^{1/2} \Ybf \Bbf \Ybf \Bbf^{1/2} 
\amp = \amp \Bbf^{1/2} \Abf \Bbf^{1/2} \\
(\Bbf^{1/2} \Zbf \Bbf^{1/2})^2 & = & \Bbf^{1/2}\Zbf \Bbf \Zbf \Bbf^{1/2} 
\amp = \amp \Bbf^{1/2} \Abf \Bbf^{1/2}
\end{eqnarray*}
imply $\Bbf^{1/2} \Ybf \Bbf^{1/2} = \Bbf^{1/2} \Zbf \Bbf^{1/2}$ by virtue of the uniqueness 
of symmetric square root. Since $\Bbf^{-1/2}$ is positive definite, $\Ybf = \Zbf$.
\end{proof}
The Cholesky factor $\Lbf$ in Lemma \ref{lemma:ricatti-equation} can be replaced by the symmetric square root of $\Bbf$. The solution, which is unique, remains the same. The Cholesky decomposition is preferred for its cheaper computational cost and better numerical stability.

Algorithm \ref{algo:varcomp-mvt-y} summarizes the MM algorithm for fitting the multi-response model \eqref{algo:varcomp-mvt-y}. Each iteration invokes $m$ Cholesky decompositions and symmetric
square roots of $d \times d$ positive definite matrices. Fortunately in most applications, $d$ is a small number. The following result guarantees the non-singularity of the Cholesky factor throughout iterations.
\begin{proposition}
\label{prop:mvt-nonzero-denom}
Assume $\Vbf_i$ has strictly positive diagonal entries. Then the symmetric matrix 
$\Mbf_i  =  (\Ibf_d \otimes \onebf_n)^T [(\onebf_d \onebf_d^T \otimes \Vbf_i) \odot \Omegabf^{-(t)}]  
(\Ibf_d \otimes \onebf_n)$ is positive definite for all $t$. Furthermore if $\Gammabf_i^{(0)} \succ \zerobf$ and no column of $\Rbf^{(t)}$ lies in the null space of $\Vbf_i$ for all $t$, then $\Gammabf_i^{(t)} \succ \zerobf$ for all $t$. 
\end{proposition}
\begin{proof}
If $\Vbf_i$ has strictly positive diagonal entries, then so does $\onebf_d \onebf_d^T \otimes \Vbf_i$,
and the Hadamard product $(\onebf_d \onebf_d^T \otimes \Vbf_i) \odot \Omegabf^{-(t)}$ is positive definite by Schur's lemma. Since the matrix $\Ibf_d \otimes \onebf_n$ has full column rank $d$, the matrix 
$\Mbf_i$ is also positive definite. Finally, if no column of $\Rbf^{(t)}$ lies in the null space of $\Vbf_i$, and $\Gammabf^{(t)}$ is positive define, then $\Gammabf_i^{(t)} \Rbf^{(t)T} \Vbf_i \Rbf^{(t)} \Gammabf_i^{(t)}$ is positive definite. The second claim follows by induction and Lemma
\ref{lemma:ricatti-equation}.
\end{proof}

\subsection*{Multivariate response, two variance components}

When there are $m=2$ variance components $\Omegabf = \Gammabf_1 \otimes \Vbf_1 + \Gammabf_2 \otimes \Vbf_2$, repeated inversion of the $nd \times nd$ covariance matrix $\Omegabf$ reduces to a single 
$nd \times nd$ simultaneous congruence decomposition and, per iteration, two $d \times d$ Cholesky decompositions and one $d \times d$ simultaneous congruence decomposition. The simultaneous congruence decomposition of the matrix pair $(\Vbf_1, \Vbf_2)$ involves generalized eigenvalues $\dbf = (d_1, \ldots, d_n)$ and a nonsingular matrix $\Ubf$ such that $\Ubf^T \Vbf_1 \Ubf = \Dbf = \diag(\dbf)$ and $\Ubf^{T} \Vbf_2 \Ubf =  \Ibf$. If the simultaneous congruence decomposition of $(\Gammabf_1^{(t)}, \Gammabf_2^{(t)})$ is $(\Lambdabf^{(t)}, \Phibf^{(t)})$ with $\Phibf^{(t)T} \Gammabf_1^{(t)} \Phibf^{(t)}= \Lambdabf^{(t)} = \diag(\lambdabf^{(t)})$ and $\Phibf^{(t)T} \Gammabf_2^{(t)} \Phibf^{(t)} = \Ibf_d$, then 
\begin{eqnarray*}
	\Omegabf^{(t)} &=& (\Phibf^{-(t)} \otimes \Ubf^{-1})^T (\Lambdabf^{(t)} \otimes \Dbf + \Ibf_d \otimes \Ibf_n) (\Phibf^{-(t)} \otimes \Ubf^{-1}) \\
	\Omegabf^{-(t)} &=& (\Phibf^{(t)} \otimes \Ubf) (\Lambdabf^{(t)} \otimes \Dbf + \Ibf_d \otimes \Ibf_n)^{-1} (\Phibf^{(t)} \otimes \Ubf)^T \\
	\det \Omegabf^{(t)} &=& \det (\Lambdabf^{(t)} \otimes \Dbf + \Ibf_d \otimes \Ibf_n) \det[ (\Phibf^{-(t)} \otimes \Ubf^{-1})^T (\Phibf^{-(t)} \otimes \Ubf^{-1})] \\
	&=& \det (\Lambdabf^{(t)} \otimes \Dbf + \Ibf_d \otimes \Ibf_n) \det (\Gammabf_2^{(t)} \otimes \Vbf_2) \\
	&=& \det (\Lambdabf^{(t)} \otimes \Dbf + \Ibf_d \otimes \Ibf_n) \det (\Gammabf_2^{(t)})^n \det(\Vbf_2)^d.
\end{eqnarray*}
Updating the fixed effects reduces to a weighted least squares problem for the transformed responses 
$\tilde \Ybf = \Ubf^T \Ybf$, transformed predictor matrix $\tilde \Xbf = \Ubf^T \Xbf$,
and observation weights $(\lambda_k^{(t)} d_i + 1)^{-1}$. Algorithm \ref{algo:2vc-mvt-y} summarizes the simplified MM algorithm. Detailed derivations are relegated to Section
\ref{sec:derive-mvt-2vc-algo}
of the Supplementary Materials.

\begin{algorithm}[ht!]
\SetKwInOut{Input}{Input}\SetKwInOut{Output}{Output}
\Input{$\Ybf$, $\Xbf$, $\Vbf_1$, $\Vbf_2$}
\Output{MLE $\widehat \Bbf$, $\widehat \Gammabf_1$, $\widehat \Gammabf_2$}
Simultaneous congruence decomposition: $(\Dbf, \Ubf) \gets (\Vbf_1, \Vbf_2)$ \;
Transform data: $\tilde \Ybf \gets \Ubf^T \Ybf$, $\tilde \Xbf \gets \Ubf^T \Xbf$ \;
Initialize $\Gammabf_1^{(0)}$, $\Gammabf_2^{(0)}$ positive definite \;
\Repeat{objective value converges}{
Simultaneous congruence decomposition $(\Lambdabf^{(t)}, \Phibf^{(t)}) \gets (\Gammabf_1^{(t)}, \Gammabf_2^{(t)})$ \;
$\Bbf^{(t)} \gets \arg \min_{\Bbf} \, [\vect (\tilde \Ybf \Phibf^{(t)}) - ( \Phibf^{(t)T} \otimes \tilde \Xbf) \vect \Bbf]^T (\Lambdabf^{(t)} \otimes \Dbf + \Ibf_d \otimes \Ibf_n)^{-1} [\vect (\tilde \Ybf \Phibf^{(t)}) - ( \Phibf^{(t)T} \otimes \tilde \Xbf) \vect \Bbf]$ \;
Cholesky $\Lbf_1^{(t)} \Lbf_1^{(t)T} \gets \Phibf^{(t)} \diag \left( \tr \left( \Dbf (\lambda_k^{(t)} \Dbf + \Ibf_n)^{-1} \right), k=1,\ldots,d \right) \Phibf^{(t)T}$ \;
Cholesky $\Lbf_2^{(t)} \Lbf_2^{(t)T} \gets 	\Phibf^{(t)} \diag \left( \tr \left( (\lambda_k^{(t)} \Dbf + \Ibf_n)^{-1} \right), k=1,\ldots,d \right) \Phibf^{(t)T}$ \;
$\Nbf_1^{(t)} \gets \Dbf^{1/2} [(\tilde \Ybf - \tilde \Xbf \Bbf^{(t)}) \Phibf^{(t)}) \oslash (\dbf \lambdabf^{(t)T} + \onebf_n \onebf_d^T)] \Lambdabf^{(t)} \Phibf^{-(t)}$ \;
$\Nbf_2^{(t)} \gets [(\tilde \Ybf - \tilde \Xbf \Bbf^{(t)}) \Phibf^{(t)}) \oslash (\dbf \lambdabf^{(t)T} + \onebf_n \onebf_d^T)] \Phibf^{-(t)}$ \;
$\Gammabf_i^{(t+1)} \gets \Lbf_i^{-(t)T} (\Lbf_i^{(t)T} \Mbf_i^{(t)T} \Mbf_i^{(t)} \Lbf_i^{(t)})^{1/2} \Lbf_i^{-(t)}$, $i=1,2$ \;
}
\caption{MM algorithm for multivariate response model 
$\Omegabf = \Gammabf_1 \otimes \Vbf_1 + \Gammabf_2 \otimes \Vbf_2$
with two variance components matrices. Note that $\oslash$ denotes a Hadamard
quotient.}
\label{algo:2vc-mvt-y}
\end{algorithm}

\subsection{Multivariate response model with missing responses}

In many applications the multivariate response model \eqref{eqn:loglik-mvt} involves missing responses. For instance, in testing multiple longitudinal traits in genetics, some trait values $y_{ij}$ may be missing due to dropped patient visits, while their genetic covariates are complete. Missing data destroys the symmetry of the log-likelihood \eqref{eqn:loglik-mvt} and complicates finding the MLE. Fortunately, MM algorithm \ref{algo:varcomp-mvt-y} easily adapts to this challenge.

The familiar EM argument \citep[Section 2.2]{McLachlan08EMBook} shows that
\begin{eqnarray}
	- \frac n2 \ln \det \Omegabf^{(t)} - \frac 12 \tr \{ \Omegabf^{-(t)} [ \vect(\Zbf^{(t)} - \Xbf \Bbf^{(t)})\vect(\Zbf^{(t)} - \Xbf \Bbf^{(t)})^T + \Cbf^{(t)}]\} \label{eqn:missing-minorizer}
\end{eqnarray}
minorizes the observed log-likelihood at the current iterate $(\Bbf^{(t)}, \Gammabf_1^{(t)}, \ldots, \Gammabf_m^{(t)})$. Here $\Zbf^{(t)}$ is the completed response matrix given the observed responses
$\Ybf_{\text{obs}}^{(t)}$ and the current parameter values. The complete data $\Ybf$ is assumed 
to be normally distributed  $N( \vect (\Xbf \Bbf^{(t)}), \Omegabf^{(t)})$. The 
block matrix $\Cbf^{(t)}$ is 0 except for a lower-right block consisting of a Schur complement. 

To maximize the surrogate \eqref{eqn:missing-minorizer}, we invoke the familiar minorization \eqref{eqn:logdet-majorization} and majorization \eqref{matrix_normal_majorization} to separate the variance components $\Gammabf_i$. At each iteration we impute missing entries by their conditional means, compute their conditional variances and covariances to supply the Schur complement, and then update the fixed effects and variance components by the explicit updates of Algorithm \ref{algo:varcomp-mvt-y}. The required conditional means and conditional variances can be conveniently obtained in the process of inverting $\Omegabf^{(t)}$ by the sweep operator of computational statistics \citep[Section 7.3]{Lange10NumAnalBook}. 

\subsection{Linear mixed model (LMM)}

The linear mixed model plays a central role in longitudinal data analysis. For the sake of
simplicity, consider the single-level LMM \citep{LairdWare82LMM,BatesPinheiro98Multilevel} for $n$ independent data clusters 
$(\ybf_i, \Xbf_i, \Zbf_i)$ with
\begin{eqnarray*}
\Ybf_i &=& \Xbf_i \betabf + \Zbf_i \gammabf_i + \epsilonbf_i, \quad i = 1,\ldots,n,
\end{eqnarray*}
where $\betabf$ is a vector of fixed effects, the $\gammabf_i \sim N(\zerobf, \Rbf_i(\thetabf))$ are independent random effects, and $\epsilonbf_i \sim N(\zerobf, \sigma^2 \Ibf_{n_i})$ captures random noise independent of $\gammabf_i$. We assume the matrices $\Zbf_i$ have full column rank. The
within-cluster covariance matrices $\Rbf_i(\thetabf)$ depend on a parameter vector $\thetabf$; 
typical choices for $\Rbf_i(\thetabf)$ impose autocorrelation, compound symmetry, or unstructured correlation.  It is clear that $\Ybf_i$ is normal with mean $\Xbf_i \betabf$, covariance 
$\Omegabf_i = \Zbf_i \Rbf_i(\thetabf) \Zbf_i^T + \sigma^2 \Ibf_{n_i}$, and log-likelihood
\begin{eqnarray*}
	L_i(\betabf, \thetabf, \sigma^2) &=&  - \frac 12 \ln \det \Omegabf_i - \frac{1}{2} (\ybf_i - \Xbf_i \betabf)^T \Omegabf_i^{-1} (\ybf_i - \Xbf_i \betabf).
\end{eqnarray*}
The next three facts about pseudo-inverses are used in deriving the MM algorithm for LMM.
\begin{lemma}
\label{lem:MPinv-product}
If $\Abf$ has full column rank and $\Bbf$ has full row rank, then $(\Abf\Bbf)^+ = \Bbf^+\Abf^+$. 
\end{lemma}
\begin{proof} Under the hypotheses, the representations $\Abf^+ = (\Abf^T\Abf)^+ \Abf^T = (\Abf^T\Abf)^{-1} \Abf^T$ and $\Bbf^{+} = \Bbf^T (\Bbf\Bbf^T)^{-1}$ are well known. The choice  
$\Bbf^+ \Abf^+ =  \Bbf^T (\Bbf\Bbf^T)^{-1}(\Abf^T\Abf)^{-1} \Abf^T$ satisfies the four equations
characterizing the pseudo-inverse of $\Abf \Bbf$.
\end{proof}

\begin{lemma}
\label{lem:MPinv-limit}
If $\Abf$ and $\Bbf$ are positive semidefinite matrices with the same range, then
\begin{eqnarray*}
\lim_{\epsilon \downarrow 0} \, (\Bbf+\epsilon \Ibf)(\Abf+\epsilon \Ibf)^{-1}
(\Bbf+\epsilon \Ibf) & = & \Bbf \Abf^+ \Bbf .
\end{eqnarray*}
\end{lemma}
\begin{proof} Suppose $\Abf$ has spectral decomposition $\sum_i \lambda_i \ubf_i \ubf_i^T$.
The matrix $\Pbf = \sum_{\lambda_i>0} \ubf_i \ubf_i^T$ projects onto
the range of $\Abf$ and therefore also projects onto the range
of $\Bbf$.  It follows that $\Pbf \Bbf = \Bbf$ and by symmetry that $\Bbf \Pbf = \Bbf$.
This allows us to write
\begin{eqnarray*}
& & (\Bbf+\epsilon \Ibf)(\Abf+\epsilon \Ibf)^{-1}(\Bbf+\epsilon \Ibf) \\
& = & 
\Bbf \Pbf (\Abf+\epsilon \Ibf)^{-1} \Pbf \Bbf 
 + \epsilon \Bbf \Pbf (\Abf+\epsilon \Ibf)^{-1}  + \epsilon (\Abf+\epsilon \Ibf)^{-1} \Pbf \Bbf+ \epsilon^2 (\Abf+\epsilon \Ibf)^{-1} .
\end{eqnarray*}
The last three of these terms vanish as $\epsilon \downarrow 0$; the first term
tends to the claimed limit. These assertions follow from the expressions
\begin{eqnarray*}
\Pbf (\Abf+\epsilon \Ibf)^{-1} \Pbf  & = & \Pbf (\Abf+\epsilon \Ibf)^{-1}
\amp = \amp (\Abf+\epsilon \Ibf)^{-1}\Pbf  \amp = \amp
\sum_{\lambda_i>0} \frac{1}{\lambda_i+\epsilon} \ubf_i\ubf_i^T
\end{eqnarray*}
and $\epsilon^2 (\Abf+\epsilon \Ibf)^{-1} =
\sum_i \frac{\epsilon^2}{\lambda_i+\epsilon}\ubf_i\ubf_i^T$.
\end{proof}

\begin{lemma}
\label{lem:same_range}
If $\Rbf$ and $\Sbf$ are positive definite matrices, and the conformable
matrix $\Zbf$ has full column rank, then the matrices
$\Zbf \Rbf \Zbf^T$ and $\Zbf \Sbf \Zbf^T$ share a
common range.
\end{lemma}
\begin{proof} 
In fact, both matrices have range equal to the range of $\Zbf$. 
The matrices $\Zbf$ and $\Zbf \Rbf^{1/2}$ clearly have the same
range. Furthermore, the matrices $\Zbf \Rbf^{1/2}$ 
and $\Zbf \Rbf^{1/2} \Rbf^{1/2}\Zbf^T$ also have the same range.
\end{proof}

The convexity of the map $(\Xbf, \Ybf) \mapsto \Xbf^T \Ybf^{-1} \Xbf$ and Lemmas 
\ref{lem:MPinv-product}, \ref{lem:MPinv-limit}, and \ref{lem:same_range}
now yield via the obvious limiting argument the majorization
 \begin{eqnarray*}
 \Omegabf^{(t)} \Omegabf^{-1} \Omegabf^{(t)} 
 & \!\! = \!\!& (\Zbf_i \Rbf_i(\thetabf^{(t)})\Zbf_i^{T}\! + \!\sigma^{2(t)}\Ibf_{n_i}) (\Zbf_i \Rbf_i(\thetabf) \Zbf_i^{T}\! + \!\sigma^{2}\Ibf_{n_i})^{-1} (\Zbf_i \Rbf_i(\thetabf^{(t)}) \Zbf_i^{T} \! + \! \sigma^{2(t)}\Ibf_{n_i}) \\
 &\!\! \preceq \!\! & (\Zbf_i \Rbf_i(\thetabf^{(t)}) \Zbf_i^{T}) (\Zbf_i \Rbf_i(\thetabf) \Zbf_i^{T})^{+} (\Zbf_i\Rbf_i(\thetabf^{(t)}) \Zbf_i^{T}) + \frac{\sigma^{4(t)}}{\sigma^2} \Ibf_{n_i} \\
 & \!\! = \!\! &[\Zbf_i \Rbf_i(\thetabf^{(t)}) \Zbf_i^{T} \Zbf_i^{T+}] \Rbf_i^{-1}(\thetabf) [\Zbf_i^{+} \Zbf_i\Rbf_i(\thetabf^{(t)}) \Zbf_i^{T}]+ \frac{\sigma^{4(t)}}{\sigma^2}\Ibf_{n_i}
 \end{eqnarray*}
In combination with the minorization \eqref{eqn:logdet-majorization}, this gives the
surrogate
\begin{eqnarray*}
	g_i(\thetabf, \sigma^2 \mid \thetabf^{(t)}, \sigma^{2(t)})
	&=& - \frac 12 \tr (\Zbf_i^T \Omegabf_i^{-(t)} \Zbf_i \Rbf_i(\thetabf)) - \frac 12 \rbf_i^{(t)T} \Rbf_i^{-1}(\thetabf) \rbf_i^{(t)} \\
	& &  - \frac{\sigma^2}{2} \tr (\Omegabf_i^{-(t)}) - \frac{\sigma^{4(t)}}{2\sigma^2} (\ybf_i - \Xbf_i \betabf^{(t)})^T \Omegabf_i^{-2(t)} (\ybf_i - \Xbf_i \betabf^{(t)}) + c^{(t)},
\end{eqnarray*}
for the log-likelihood $L_i(\thetabf, \sigma^2)$,
where 
\begin{eqnarray*}
	\rbf_i^{(t)} & = & (\Zbf_i^{+} \Zbf_i\Rbf_i(\thetabf^{(t)}) \Zbf_i^{T}) \Omegabf_i^{-(t)} (\ybf_i - \Xbf_i \betabf^{(t)}) \amp = \amp \Rbf_i(\thetabf^{(t)}) \Zbf_i^{T} \Omegabf_i^{-(t)} (\ybf_i - \Xbf_i \betabf^{(t)}).
\end{eqnarray*}
The parameters $\thetabf$ and $\sigma^2$ are nicely separated. To maximize the overall minorization function $\sum_i g_i(\thetabf, \sigma^2 \mid \thetabf^{(t)}, \sigma^{2(t)})$, we update $\sigma^2$ via
\begin{eqnarray*}
	\sigma^{2(t+1)} & = & \sigma^{2(t)} \sqrt{\frac{\sum_i (\ybf_i - \Xbf_i \betabf^{(t)})^T \Omegabf_i^{-2(t)} (\ybf_i - \Xbf_i \betabf^{(t)})}{\sum_i \tr (\Omegabf_i^{-(t)})}}.
\end{eqnarray*}
For structured models such as autocorrelation and compound symmetry, updating $\thetabf$ is a 
low-dimensional optimization problem that can be approached through the stationarity condition
\begin{eqnarray*}
	\sum_i \vect \left(\Zbf_i^T \Omegabf_i^{(t)} \Zbf_i - \Rbf_i^{-1}(\thetabf) \rbf_{i}^{(t)}\rbf_{i}^{(t)T}  \Rbf_i^{-1}(\thetabf) \right)^T \frac{\partial }{\partial \theta_j} \vect \,\Rbf_i(\thetabf) & = & 0
\end{eqnarray*}
for each component $\theta_j$. For the unstructured model with $\Rbf_i(\thetabf) = \Rbf$ for all $i$, 
the stationarity condition reads 
\begin{eqnarray*}
	\sum_i \Zbf_i^T \Omegabf_i^{(t)} \Zbf_i & = & \Rbf^{-1} \left( \sum_i \rbf_{i}^{(t)}\rbf_{i}^{(t)T} \right) \Rbf^{-1}
\end{eqnarray*}
and admits an explicit solution based on Lemma \ref{lemma:ricatti-equation}.

Similar tactics apply to a multilevel LMM \citep{BatesPinheiro98Multilevel} with responses
\begin{eqnarray*}
	\Ybf_i & = & \Xbf_i \betabf + \Zbf_{i1} \gammabf_{i1} + \cdots \Zbf_{im} \gammabf_{im} + \epsilonbf_i.
\end{eqnarray*}
Minorization separates parameters for each level (variance component). Depending on the
complexity of the covariance matrices, maximization of the surrogate can be accomplished
analytically. For the sake of brevity, details are omitted.

\subsection{MAP estimation}

Suppose $\betabf$ follows an improper flat prior, the variance components $\sigma_i^2$ follow  inverse gamma priors with shapes $\alpha_i>0$ and scales $\gamma_i > 0$, and these
priors are independent. The log-posterior density then reduces to
\begin{eqnarray}
&& \ln f(\betabf, \sigmabf^2 | \ybf, \Xbf)  \nonumber \\
&=& - \frac 12 \ln \det \Omegabf - \frac{1}{2} (\ybf - \Xbf \betabf)^T \Omegabf^{-1} (\ybf - \Xbf \betabf) - \sum_{i = 1}^m (\alpha_i + 1)\ln \sigma_i^2 - \sum_{i=1}^m \frac{\gamma_i}{\sigma_i^2} + c,  \label{eqn:vc-posterior}
\end{eqnarray}
where $c$ is an irrelevant constant. The MAP estimator of 
$(\betabf, \sigmabf^2)$ is the mode of the posterior distribution. The update 
(\ref{generalized_least_squares}) of $\betabf$ given $\sigmabf^2$ remains the same. To update $\sigmabf^2$ given $\betabf$, apply the same minorizations \eqref{eqn:matrix-ineq} and \eqref{eqn:logdet-majorization} to the first first two terms of equation \eqref{eqn:vc-posterior}.
This separates parameters and yields a convex surrogate for each $\sigma_i^2$. 
The minimum of the $\sigma_i^2$ surrogate is defined by the stationarity condition
\begin{eqnarray*}
0 & = & -\frac{1}{2} \tr(\Omegabf^{-(t)}\Vbf_i)+\frac{\sigma_i^{4(t)}}{2\sigma_i^{4}}
(\ybf - \Xbf \betabf^{(t)})^T \Omegabf^{-(t)} \Vbf_i \Omegabf^{-(t)} (\ybf - \Xbf \betabf^{(t)})
-\frac{\alpha_i+1}{\sigma_i^2}+\frac{\gamma_i}{\sigma_i^4}.
\end{eqnarray*}
Multiplying this by $\sigma_i^4$ gives a quadratic equation in $\sigma_i^2$.
The positive root should be taken to meet the nonnegativity constraint on $\sigma_i^2$.

For the multivariate response model \eqref{eqn:loglik-mvt}, we assume the variance components $\Gammabf_i$ follow independent inverse Wishart distributions with degrees of freedom $\nu_i > d-1$ and scale matrix $\Psibf_i \succ \zerobf$. The log density of the posterior distribution is
\begin{eqnarray}
L(\Bbf, \Gammabf|\Xbf, \Ybf) & =&  - \frac 12 \ln \det \Omegabf - \frac{1}{2} \vect (\Ybf - \Xbf \Bbf)^T \Omegabf^{-1} \vect (\Ybf - \Xbf \Bbf) \nonumber \\
&  & - \frac 12 \sum_{i=1}^m (\nu_i+d+1) \ln \det \Gammabf_i - \frac 12 \sum_{i=1}^m \tr(\Psibf_i\Gammabf_i^{-1}) + c, \label{eqn:pos-loglik-mvt}
\end{eqnarray}
where $c$ is an irrelevant constant. Invoking the minorizations \eqref{eqn:logdet-majorization} and \eqref{matrix_normal_majorization} for the first two terms and the supporting hyperplane minorization
\begin{eqnarray*}
-\ln\det \Gammabf_i &\geq& -\ln\det \Gammabf_i^{(t)} - \tr\{\Gammabf_i^{-(t)}(\Gammabf_i - \Gammabf_i^{(t)})\} 
\end{eqnarray*}
for  $-\ln \det \Gammabf_i$ gives the surrogate function
\begin{eqnarray*}
	g(\Gammabf|\Gammabf^{(t)})
	&=& - \frac 12 \sum_{i=1}^m \tr (\Omegabf^{-(t)} (\Gammabf_i \otimes \Vbf_i)) - \frac 12 \sum_{i=1}^m \tr \left( \Gammabf_i^{(t)} \Rbf^{(t)T} \Vbf_i \Rbf^{(t)} \Gammabf_i^{(t)} \Gammabf_i^{-1} \right)  \\
    & & - \frac 12 \sum_{i=1}^m (\nu_i+d+1) \tr(\Gammabf_i^{-(t)}\Gammabf_i) - \frac 12 \sum_{i=1}^m \tr(\Psibf_i\Gammabf_i^{-1})+ c^{(t)}.
\end{eqnarray*}
The optimal $\Gammabf_i$ satisfies the stationarity condition
\begin{eqnarray*}
&  & (\Ibf_d \otimes \onebf_n)^T [(\onebf_d \onebf_d^T \otimes \Vbf_i) \odot \Omegabf^{-(t)}]  (\Ibf_d \otimes \onebf_n) + (\nu_i+d+1)\Gammabf_i^{-(t)} \\ 
& = & \Gammabf_i^{-1} (\Gammabf_i^{(t)} \Rbf^{(t)T} \Vbf_i \Rbf^{(t)} \Gammabf_i^{(t)}+\Psibf_i) \Gammabf_i^{-1} \label{eqn:mvt-stationarity-bayes}
\end{eqnarray*}
and can be found using Lemma \ref{lemma:ricatti-equation}.

\subsection{Variable selection}

In the statistical analysis of high-dimensional data, the imposition of sparsity leads to better interpretation and more stable parameter estimation. MM algorithms mesh well with penalized estimation.  The simple variance components model \eqref{eqn:vc-loglike} illustrates this fact.
For the selection of fixed effects, minimizing the lasso-penalized log-likelihood
\begin{eqnarray*}
	- L(\betabf, \sigmabf^2) + \lambda \sum_j |\beta_j| 
\end{eqnarray*}
is often recommended \citep{SchelldorferBuhlmannVandeGeer11L1LMM}. The only change to the MM Algorithm \ref{algo:varcomp-univariate-y} is that in estimating $\betabf$, one solves a lasso penalized general least squares problem rather than an ordinary general least squares problem. 
The updates of the variance components $\sigma_i^2$ remain the same. For selection among a large number of variance components, one can minimize the ridge-penalized log-likelihood
\begin{eqnarray*}
	- L(\betabf, \sigmabf^2) + \lambda \sum_{i=1}^m \sigma_i^2 \label{eqn:vc-ridge}
\end{eqnarray*}
subject to the nonnegativity constraints $\sigma_i^2 \ge 0$. Here the standard deviations 
$\sigma_i$ are the underlying parameters. The variance update \eqref{eqn:MM-update-univariate} becomes
\begin{eqnarray*}
	\sigma_i^{2(t+1)} & = & \sigma_i^{2(t)} \sqrt{\frac{(\ybf - \Xbf \betabf^{(t)})^T \Omegabf^{-(t)} \Vbf_i \Omegabf^{-(t)} (\ybf - \Xbf \betabf^{(t)})}{\tr (\Omegabf^{-(t)}  \Vbf_i) + 2 \lambda}}, \quad i = 1,\ldots, m . \label{eqn:lasso-mm-update}
\end{eqnarray*}
The updates for the fixed effects $\betabf$ are unaffected. Equation \eqref{eqn:lasso-mm-update}
clearly exhibits shrinkage but no thresholding. The lasso penalized log-likelihood
\begin{eqnarray}
	- L(\betabf, \sigmabf^2) + \lambda \sum_{i=1}^m \sigma_i \label{eqn:vc-lasso}
\end{eqnarray}
subject to nonnegativity constraint $\sigma_i \ge 0$ achieves both ends. The update of $\sigma_i$ is chosen among the positive roots of a quartic equation and the boundary 0, whichever yields a lower objective value. 

\subsection{Beyond the linear model}

One can extend the MM algorithms to binary and discrete response data with the framework of generalized estimating equations (GEE) \citep{Liang1986}. Again consider $n$ independent data clusters $(\ybf_i, \Xbf_i)$. In longitudinal studies, $(\ybf_i, \Xbf_i)$ would be the responses and clinical covariates of subject $i$ at different time points. In genetic studies, $(\ybf_i, \Xbf_i)$ would be the trait values and covariates of individuals within family $i$. GEE captures the within-cluster correlation by specifying the first two moments of the conditional distribution of $\ybf_i$ given $\Xbf_i$, namely $\mu_{ij} = E(y_{ij} | \xbf_{ij})$ and $\sigma^2_{ij}=\text{Var}(y_{ij} | \xbf_{ij})$. If one assumes that $y_{ij}$ follows an exponential family with canonical link, then 
\begin{eqnarray*}
\mu_{ij}(\betabf) & = & \mu(\theta_{ij}) \;\; \textrm{ and } \;\;
\sigma^2_{ij}(\betabf) \amp = \amp \phi \mu'(\theta_{ij}), \quad i=1, \ldots, n, \; j=1, \ldots, n_i, 
\end{eqnarray*}
where $\mu(t)$  is a differentiable canonical link function, $\mu'(t)$ is its first derivative, $\theta_{ij} = \xbf_{ij}^T \betabf$  is the linear systematic part of $y_{ij}$ associated with the covariates, $\phi$ is an over-dispersion parameter, and $\betabf$ is the vector of fixed effects.

The GEE estimator of $\betabf$ solves the equation
\begin{eqnarray*}
\sum_{i=1}^n d\mubf_i(\betabf)^T \Vbf_i^{-1}(\alphabf,\betabf) [ \ybf_i - \mubf_i(\betabf) ] 
& = & \zerobf_p,
\end{eqnarray*}
where $\ybf_i = (y_{i1}, \dots, y_{in_i})^T$, $\mubf_i(\betabf) = [\mu_{i1}(\betabf), \ldots, \mu_{in_i}(\betabf)]^T$, $\Vbf_i=\text{cov}(\ybf_i)$ is the working covariance matrix of the $i$-th subject, and $d\mubf_i(\betabf)$ is the differential of $\mubf_i(\betabf)$.  In longitudinal studies, $\Vbf_i$ is often parameterized as $\Vbf_i (\alphabf, \betabf) = \Abf_i^{1/2}(\betabf) \Rbf_i(\alphabf) \Abf_i^{1/2}(\betabf)$, where $\Abf^{1/2}(\betabf)$ is a diagonal matrix with standard deviations $\sigma_{ij}$ along its diagonal, and $\Rbf(\alphabf)$ is a correlation matrix with parameters $\alphabf$. This parameterization is too restrictive in many other applications. For instance, in genetic studies, it is critical to dissect the variance into different sources such as additive, dominance, and household environment \citep{Lange02GeneticsBook}. This suggests the variance component parameterization
\begin{eqnarray*}
	\Rbf_i(\sigmabf^2) & = & \sigma_a^2 2\Phibf_i + \sigma_d^2 \Deltabf_{7,i} + \sigma_{h}^2 \Hbf_i + \sigma_e^2 \Ibf_{n_i}, \quad \sigma_a^2 + \sigma_d^2 + \sigma_h^2 + \sigma_e^2 \amp = \amp 1,
\end{eqnarray*}
where in the $i$-th family $\Phibf_i$ is the theoretical kinship matrix, $\Deltabf_{7,i}$ is the dominance variance matrix, and $\Hbf_i$ is the household indicator matrix. The matrices
$\Phibf_i$, $\Deltabf_{7,i}$, and $\Hbf_i$ are correlation matrices, and the simplex constraint ensures $\Rbf_i$ is as well. In general, the variance component parameterization $\Rbf_i (\sigmabf^2) = \sum_{j=1}^m \sigma_j^2 \Rbf_{ij}$ with the simplex constraint in force is reasonable. In this setting the GEE update of $\betabf$ given $\sigmabf^2$ solves the equation 
\begin{eqnarray*}
\sum_{i=1}^n \Xbf_i^T \Abf_i^{-1/2}(\betabf) \Rbf_i^{-1}(\sigmabf^2) \Abf^{-1/2}(\betabf) [ \ybf_i - \mubf_i(\betabf) ] & = & \zerobf_p.
\end{eqnarray*}
This is just the classical GEE update. The difficulty lies in updating $\sigmabf^2$ given $\betabf$. We propose minimizing the sum $\sum_{i=1}^n \psi(\Rbf_i)$, where $\psi(t)$ is a scalar convex loss function. Example loss functions include the Mahalanobis criterion
\begin{eqnarray*}
	\sum_{i=1}^n \psi_1(\Rbf_i) &=& \sum_{i=1}^n (\ybf_i - \widehat \mubf_i)^T \widehat \Abf_i^{-1/2} \Rbf_i^{-1}(\sigmabf^2) \widehat \Abf_i^{-1/2} (\ybf_i - \widehat \mubf_i)
\end{eqnarray*}
and the sum of squared Frobenius distances
\begin{eqnarray*}
	\sum_{i=1}^n \psi_2(\Rbf_i) &=& \sum_{i=1}^n \tr \left[ (\ybf_i - \widehat \mubf_i) (\ybf_i - \widehat \mubf_i)^T - \widehat \Abf_i^{1/2} \Rbf_i(\sigmabf^2) \widehat \Abf_i^{1/2} \right]^2.
\end{eqnarray*}
The convexity of $\psi_1(t)$ entails a minorization similar to the minorization 
\eqref{eqn:matrix-ineq}.  Minimizing the surrogate then yields the MM update
\begin{eqnarray*}
	\sigma_j^{2(t+1)} = \frac{\sum_{i=1}^n (\ybf_i - \widehat \mubf_i)^T \widehat \Abf_i^{-1/2} \Rbf_i^{-1}(\sigmabf^{2(t)}) \Rbf_{ij} \Rbf_i^{-1}(\sigmabf^{2(t)}) \widehat \Abf_i^{-1/2} (\ybf_i - \widehat \mubf_i)}{\sum_{j=1}^m \sum_{i=1}^n (\ybf_i - \widehat \mubf_i)^T \widehat \Abf_i^{-1/2} \Rbf_i^{-1}(\sigmabf^{2(t)}) \Rbf_{ij} \Rbf_i^{-1}(\sigmabf^{2(t)}) \widehat \Abf_i^{-1/2} (\ybf_i - \widehat \mubf_i)}.
\end{eqnarray*}
Under $\psi_2(t)$ the MM update boils down to projection onto the simplex. Further exploration of these ideas probably deserves another paper and will be omitted here for the sake of brevity.

\section{A numerical example}
\label{sec:numerical-example}

Quantitative trait loci (QTL) mapping aims to identify genes associated with a quantitative trait. Current sequencing technology measures millions of genetic markers in study subjects. Traditional single-marker tests suffer from low power due to the low frequency of many markers and the 
corrections needed for multiple hypothesis testing. Region-based association tests are a powerful alternative for analyzing next generation sequencing data with abundant rare variants.

Suppose $\ybf$ is a $n \times 1$ vector of quantitative trait measurements on $n$ people, $\Xbf$ is an $n \times p$ predictor matrix (incorporating predictors such as sex, smoking history, and principal components for ethnic admixture), and $\Gbf$ is an $n \times m$ genotype matrix of $m$ genetic variants in a pre-defined region. The linear mixed model assumes 
\begin{eqnarray*}
	\Ybf &=& \Xbf \betabf + \Gbf \gammabf + \epsilonbf, \quad \gammabf \sim N(\zerobf,\sigma_g^2\Ibf), \quad \epsilonbf \sim N(\zerobf,\sigma_e^2\Ibf_n),
\end{eqnarray*}
where $\betabf$ are fixed effects, $\gammabf$ are random genetic effects, and $\sigma_g^2$ and $\sigma_e^2$ are variance components for the genetic and environmental effects, respectively. Thus, the phenotype vector $\Ybf$ has covariance $\sigma_g^2 \Gbf \Gbf^T + \sigma_e^2\Ibf_n$, where $\Gbf \Gbf^T$ is the kernel matrix capturing the overall effect of the $m$ variants. Current approaches test the null hypothesis $\sigma_g^2=0$ for each region separately and then adjust for multiple testing \citep{Lee14RVSurvey}. Hence, we consider the joint model
\begin{eqnarray*}
	\ybf &=& \Xbf \betabf + s_1^{-1/2} \Gbf_1 \gammabf_1 + \cdots + s_m^{-1/2}\Gbf_m \gammabf_m + \epsilonbf, \label{vc-model} \\
	\gammabf_i &\sim& N(\zerobf,\sigma_{i}^2 \Ibf), \quad \epsilonbf \sim N(\zerobf,\sigma_e^2\Ibf_n)
\end{eqnarray*}
and select the variance components $\sigma_i^2$ via the penalization \eqref{eqn:vc-lasso}. Here $s_i$ is the number of variants in region $i$, and the weights $s_i^{-1/2}$ put all variance components on the same scale.

We illustrate this approach using the COPDGene exome sequencing study (\url{http://www.copdgene.org/}) \citep{Regan10COPD}. After quality control, $399$ individuals and 646,125 genetic variants remain for analysis. Genetic variants are grouped into 16,619 genes to expose those genes associated with the complex trait {\tt height}. We include {\tt age}, {\tt sex}, and the top 3 principal components in the mean effects. Because the number of genes vastly exceeds the sample size $n=399$,  we first pare the 16,619 genes down to 200 genes according to their marginal likelihood ratio test p-values and then carry out penalized estimation of the 200 variance components in the joint model \eqref{eqn:vc-lasso}. This is similar to the sure independence screening strategy for selecting mean effects \citep{FanLv2008}. Genes are ranked according to the order they appear in the lasso solution path. Table \ref{table:lasso-rank} lists the top 10 genes together with their marginal LRT p-values. Figure \ref{fig:copd-solpath} displays the corresponding segment of the lasso solution path. It is noteworthy that the ranking of genes by penalized estimation differs from the ranking according to marginal p-values. The same phenomenon occurs in selection of highly correlated mean predictors. This penalization approach for selecting variance components warrants further theoretical study. It is reassuring that the simple MM algorithm scales to high-dimensional problems.

\begin{table}
\begin{center}
{\small
\begin{tabular}{crrr}
\toprule
Lasso Rank & Gene & Marginal P-value & \# Variants \\
\midrule
1 & DOLPP1 & $2.35 \times 10^{-6}$ & 2 \\ 
2 & C9orf21 & $3.70 \times 10^{-5}$ & 4 \\
3 & PLS1 & $ 2.29\times 10^{-3}$ & 5 \\
4 & ATP5D & $6.80 \times 10^{-7}$ & 3  \\
5 & ADCY4 & $1.01 \times 10^{-3}$ & 11 \\
6 & SLC22A25 & $3.95 \times 10^{-3}$ & 14 \\
7 & RCSD1 & $9.04 \times 10^{-4}$ & 4 \\
8 & PCDH7 & $1.20 \times 10^{-4}$ & 7 \\
9 & AVIL& $8.34 \times 10^{-4}$ & 11\\
10 & AHR & $1.14 \times 10^{-3}$ & 7 \\
\bottomrule
\end{tabular}
}
\caption{Top 10 genes selected by the lasso penalized variance component model \eqref{eqn:vc-lasso} in an association study of 200 genes and the complex trait {\tt height}.}
\end{center}
\label{table:lasso-rank}
\end{table}

\begin{figure}
\begin{center}
	\includegraphics[width=4.5in]{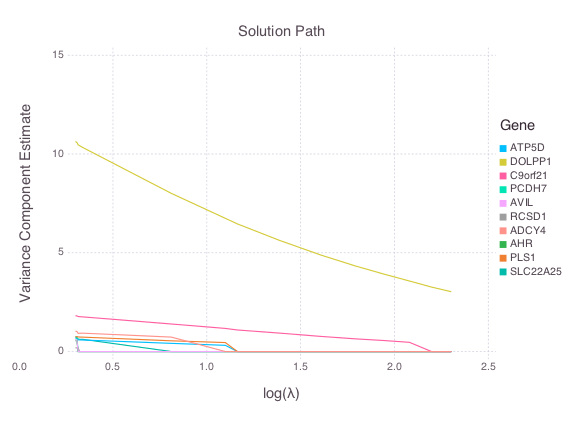}
\end{center}
\caption{Solution path of the lasso penalized variance component model \eqref{eqn:vc-lasso} in an association study of 200 genes and the complex trait {\tt height}.}
\label{fig:copd-solpath}
\end{figure}

%

\section{Discussion}

The current paper leverages the MM principle to design powerful and versatile algorithms
for variance components estimation. The MM algorithms derived are notable for their simplicity, generality, numerical efficiency, and theoretical guarantees. Both ordinary MLE and REML are apt to benefit. Other extensions are possible. In nonlinear models \citep{BatesWatts88NonlinearBook,LindstromBates90NonlinearLMM}, the mean response is a nonlinear function in the fixed effects $\betabf$. One can easily modify the MM algorithms to update $\betabf$ by a few rounds of Gauss-Newton iteration. The variance components updates remain unchanged.

One can also extend our MM algorithms to elliptically symmetric densities
\begin{eqnarray*}
f(\ybf) & = & \frac{e^{-\frac{1}{2}\kappa(\delta^2)}}
{(2 \pi)^{\frac{n}{2}} (\det\Omegabf)^{\frac{1}{2}}}  \label{elsymden}
\end{eqnarray*}
defined for $\ybf \in \mathbb{R}^n$, where 
$\delta^2 = (\ybf - \mubf)^T \Omegabf^{-1}(\ybf - \mubf)$ 
denotes the Mahalanobis distance between $\ybf$ and $\mubf$.
Here we assume that the function $\kappa(s)$ is strictly increasing and strictly 
concave. Examples of elliptically symmetric densities include the multivariate $t$, slash, 
contaminated normal, power exponential, and stable families.
Previous work \citep{HuberRonchetti09RobustStatistics,Lange1993normal} has focused on using the MM principle
to convert parameter estimation for these robust families into parameter estimation 
under the multivariate normal.  One can chain the relevant majorization 
$\kappa(s) \le \kappa(s^{(t)})+\kappa'(s^{(t)})(s-s^{(t)})$ with our
previous minorizations and simultaneously split variance components and
pass to the more benign setting of the multivariate normal.

\section*{Acknowledgments}
The research is partially supported by NSF grant DMS-1055210 and NIH grants R01 HG006139, R01 GM53275, and R01 GM105785. The authors thank Michael Cho, Dandi Qiao, and Edwin Silverman for their assistance in processing and assessing COPDGene exome sequencing data. COPDGene is supported by NIH R01 HL089897 and R01 HL089856.


\bibliography{vc_notes}
\bibliographystyle{apalike}

\end{document}